\documentclass[a4paper,10pt]{article}

\usepackage[top=1in, bottom=1.25in, left=1.1in, right=1.1in]{geometry}
\usepackage{graphicx}
\usepackage{authblk}
\usepackage[round]{natbib}
\usepackage{hyperref}
\hypersetup{
     colorlinks   = true,
     citecolor    = blue
}
\usepackage{physics}
\usepackage{lipsum}	
\usepackage{amsthm}
\usepackage{booktabs} 
\usepackage{caption}
\usepackage{amssymb,amsfonts,amsthm}
\usepackage{mathtools}
\usepackage{tikz}
\usetikzlibrary{decorations.pathreplacing}
\usetikzlibrary{patterns}
\usepackage{pifont}
\usepackage{multirow}
\usepackage{algorithm}
\usepackage{algpseudocode}
\usepackage{hhline}
\usepackage{todonotes}

\newtheorem{theorem}{Theorem}[section]
\newtheorem{lemma}[theorem]{Lemma}
\newtheorem{proposition}[theorem]{Proposition}

\newtheorem{definition}[theorem]{Definition}

\begin{document}

\title{Scalable Bayesian Inference for \\ Population Markov Jump Processes}
\author{Iker Perez}
\author{Theodore Kypraios}
\affil{School of Mathematical Sciences, University of Nottingham}
\date{\vspace{-10pt}}

\maketitle

\begin{abstract}
Bayesian inference for Markov jump processes (MJPs) where available observations relate to either system states or jumps typically relies on data-augmentation Markov Chain Monte Carlo. State-of-the-art developments involve representing MJP paths with auxiliary candidate jump times that are later thinned. However, these algorithms are i) unfeasible in situations involving large or infinite capacity systems and ii) not amenable for all observation types. In this paper we establish and present a general data-augmentation framework for population MJPs  based on uniformized representations of the underlying non-stationary jump processes. This leads to multiple novel MCMC samplers which enable \textit{exact} (in the Monte Carlo sense) inference tasks for model parameters. We show that proposed samplers outperform existing popular approaches, and offer substantial efficiency gains in applications to partially observed stochastic epidemics, immigration processes and predator-prey dynamical systems.
\end{abstract}

\section{Introduction} 

Population \textit{Markov jump processes} (MJPs) are stochastic processes whose dynamics underpin many observable phenomena, in diverse fields such as stochastic epidemic \citep{o1999bayesian}, immigration-death systems \citep{cappe2003reversible,zhang2017efficient}, chemical/molecular models \citep{hobolth2009,georgoulas2017unbiased} or queueing systems \citep{sutton2011,Perez2017}, to name only a few. In this paper, we present a novel and general framework for designing scalable auxiliary-variable data-augmentation algorithms, which allow for  \textit{exact} (in Monte Carlo sense) Bayesian inference for MJPs. In contrast to the state-of-the-art method of \citet{rao13a}, our framework provides us with a class of Markov chain Monte Carlo (MCMC) algorithms that are amenable to all relevant application fields, where observations may consist of either population counts or jumps. Furthermore, we present efficient algorithmic designs that address augmentation and inferential tasks in systems where the population is large. We demonstrate this by reporting on substantial efficiency and scalability gains, in application to partially-observed birth-death, (stochastic) epidemic  and predator-prey models. This general framework is built over \textit{uniformized} representations of non-stationary jump processes \cite[cf.][]{jensen1953markoff,van1992uniformization}, and we show that algorithms presented in  \cite{rao2012mcmc,rao13a} are derived as special cases.

\subsection{Jump processes}

An MJP is a pure-jump right-continuous stochastic process $X=(X_t)_{t\geq 0}$, such that time-indexed variables $X_t$ are defined within some measurable space $(\mathcal{S},\Sigma_\mathcal{S})$. Here, $\mathcal{S}$ is a countable set of states, and $\Sigma_\mathcal{S}$ stands for its power set. The process $X$ is governed by a sequence of \textit{intensity matrices} $\boldsymbol{Q}=\{Q(t) \, : \, t\geq 0\}$, indexed over time, so that
$$\mathbb{P}(X_{t+\mathrm{d}t}=x'|X_{t}=x) = \mathbb{I}_{(x=x')} + Q_{x,x'}(t)\mathrm{d}t + o(\mathrm{d}t)$$
for all $x,x'\in\mathcal{S}$ and $t\geq 0$; where $\mathbb{I}_{(\cdot)}$ defines a logical indicator function. Hence, elements of $Q(t)$ describe the rates for \textit{jumps} between states at time $t\geq 0$, and each $Q_{x,x'}(t)$, $x,x'\in\mathcal{S}$, is an \textit{intensity} function over time. Finally, $0\leq Q_{x,x'}(t)< \infty$ for all $x \neq x'$ and $Q_x(t) \coloneqq Q_{x,x}(t) = - \sum_{x'\in\mathcal{S}: x\neq x'} Q_{x,x'}(t)$ so that the various rows sum to $0$.

\noindent \textbf{Piecewise-constant representation.} An MJP is further characterized by a \textit{path} or \textit{trajectory} $(\boldsymbol{t},\boldsymbol{x})$, where $\boldsymbol{t}=\{t_0,\dots,t_{n}\}$ denotes a sequence of transition times, s.t. $t_0=0$, and $\boldsymbol{x}=\{x_0,\dots,x_{n}\}$ are the corresponding states. Over a time interval $[0,T]$, a process $X\equiv (\boldsymbol{t},\boldsymbol{x})$ is a random variable on a measurable space $(\mathcal{X},\Sigma_\mathcal{X})$ supporting finite $\mathcal{S}$-valued trajectories. On a basic level, $\mathcal{X}=\cup_{i=0}^\infty ([0,T]\times\mathcal{S})^i$ and the collection $\Sigma_\mathcal{X}$ stands for the corresponding union $\sigma$-algebra. This allows the assignment of a dominating base measure $\mu_\mathcal{X}$ w.r.t which define a trajectory density
\begin{align}
f_{X}(\boldsymbol{t},\boldsymbol{x}|\boldsymbol{Q})  & =  \pi(x_0) e^{\int_{t_n}^T Q_{x_n}(s)\mathrm{d}s} \prod_{i=1}^n Q_{x_{i-1},x_{i}}(t_i) e^{\int_{t_{i-1}}^{t_i}Q_{x_{i-1}}(s)\mathrm{d}s}, \label{pathProbs}
\end{align}
where $\pi(\cdot)$ is the distribution assigned (over $\mathcal{S}$) to the starting value. Noticeably, the time to departure or \textit{jump} from any state $x\in\mathcal{S}$, regardless of the destination, is driven by a density
\begin{align*}
f_{t_{i+1}}(t|x_i,t_i) = Q_{x_i}(t) e^{-\int_{t_i}^{t}Q_{x_i}(s)\mathrm{d}s}, \quad t>t_i, \, i=1,\dots,n-1.
\end{align*}
Thus, inter-arrival times in $\boldsymbol{t}$ are linked to diagonal elements of $\{Q(t) \, : \, t\geq 0\}$, and $Q_{x}(t+s)$, $s>0$ is often referred to as the \textit{hazard} function to the origin state $(t,x)\in[0,T]\times\mathcal{S}$. Finally, transitions in $\boldsymbol{x}$ are proportional to off-diagonal elements, s.t. $\mathbb{P}(x_{i+1}=x|x_i,t_i,t_{i+1})=Q_{x_i,x}(t_{i+1})/|Q_{x_i}(t_{i+1})|$. For details, we refer the reader to \cite{daley2007introduction}.

\noindent \textbf{Stationary models.} If $X$ is assumed to be a time-homogeneous process, ignoring seasonal effects and thus governed by a \textit{generator matrix} $Q(t)\equiv Q, t\geq 0$; then, inter-arrival times in $\boldsymbol{t}$ are exponentially distributed random variables and $(\boldsymbol{t},\boldsymbol{x})$ is a (Markov) renewal process.

\subsection{Population models and Bayesian inferential tasks}

Throughout this paper, a Markov \textit{population model} is represented by a non-stationary MJP whose support space $\mathcal{S}$ is \textit{countable} and possibly infinite. Matrices $Q(t), t\geq 0$ are assumed to be \textit{sparse} and parametrized by some arbitrary vector of independent rates $\boldsymbol{\lambda}$, which scale along with levels of populations in $X$. An upper-bound over a sequence of matrices $\boldsymbol{Q}\equiv\boldsymbol{Q}(\boldsymbol{\lambda})$ may take extraordinarily large values. 

\noindent \textbf{Bayesian inferential task.} Let $\boldsymbol{O}=\{O_r\}_{r\geq 1}$ denote some observations at arbitrary (ordered) time points $t_r\in[0,T]$, $r\geq 1$, which relate to a population model realization $X$ with unknown matrices $\boldsymbol{Q}(\boldsymbol{\lambda})$. The basis for inference on the (unknown) vector $\boldsymbol{\lambda}$ is a density or mass $\mathcal{L}(\boldsymbol{O}|X)$ for the observation model; and posterior rate densities are proportional to an infinite weighted product of MJP path densities $X$ in \eqref{pathProbs}, i.e.
\begin{align}
f_{\boldsymbol{Q}}(\boldsymbol{\lambda}|\boldsymbol{O}) \propto f_{\boldsymbol{Q}}(\boldsymbol{\lambda})  \cdot \int_{\mathcal{X}} \mathcal{L}(\boldsymbol{O}|\boldsymbol{t},\boldsymbol{x})\, f_{X}(\boldsymbol{t},\boldsymbol{x}|\boldsymbol{Q}(\boldsymbol{\lambda}))\, \mu_{\mathcal{X}}(d\boldsymbol{t},d\boldsymbol{x}), \label{intractLikel}
\end{align}
where $f_{\boldsymbol{Q}}(\boldsymbol{\lambda})$ defines a prior over the rates. This is an analytically, and often computationally, intractable expression. It is hard to design a generic framework to perform {\em exact} Monte Carlo inference, yet remain adaptable to any type of jump process $X$ and observation model $\mathcal{L}(\boldsymbol{O}|X)$. Consequently, many solutions either focus on approximate inferential methods, or are limited to homogeneous systems and address constrained biological models where population measurements must be subject to observation noise. Such approaches can lead to computationally efficient methods by relying on simplifying independence assumptions \cite{opper2008variational}, diffusions with continuous support \cite{golightly2015bayesian} or linear noise approximations \cite{golightly2018efficient}. 

\subsection{Exact inference and Markov chain Monte Carlo}

Exact inference often proceeds by MCMC, and alternates sampling between the latent process $(\boldsymbol{t},\boldsymbol{x})$ and rates $\boldsymbol{\lambda}$. Thus, it is concerned with the joint density $f_{\boldsymbol{Q},X}(\boldsymbol{\lambda},\boldsymbol{t},\boldsymbol{x}|\boldsymbol{O})$, and entails data augmentation procedures from a conditional
\begin{align}
f_{X}(\boldsymbol{t},\boldsymbol{x}|\boldsymbol{\lambda},\boldsymbol{O})  \propto \mathcal{L}(\boldsymbol{O}|\boldsymbol{t},\boldsymbol{x}) f_{X}(\boldsymbol{t},\boldsymbol{x}|\boldsymbol{Q}(\boldsymbol{\lambda})) , \label{targetdens}
\end{align}
which may take multiple forms based on observation model dependencies for $\boldsymbol{O}|\boldsymbol{t},\boldsymbol{x}$. In every instance, sampling a trajectory from \eqref{targetdens} brings about substantial tractability challenges; there can exist infinitely many jumps and we require to explore transitions across large or infinite subsets of $\mathcal{S}$. To allow for a generic adaptable algorithm design, sampling commonly proceeds by means of blocked (Poisson) \textit{thinning} procedures, in summary 
\begin{itemize}
\item a set of \textit{candidate} jump times $\hat{\boldsymbol{t}}|\boldsymbol{O},\boldsymbol{Q}(\boldsymbol{\lambda})$ is first produced, with some \textit{conditional} intensity process $\Omega(t), t\geq 0$, and s.t.  every $\Omega(t)$ \textit{dominates} all diagonal elements of $Q(t)$,
\item an \textit{augmented} sequence $\hat{\boldsymbol{x}}|\boldsymbol{t},\boldsymbol{O},\boldsymbol{Q}(\boldsymbol{\lambda})$ is sampled from an appropriate forward-backward algorithm; this must allow for self transitions and thus \textit{thin} a portion of candidate jump times.
\end{itemize}  
Within time-homogeneous jump systems, such procedures may be supported on matrix exponential representations for transition probabilities \citep[see][]{fearnhead2006exact} or, ideally, built over \textit{uniformization} alternatives and the seminal contributions of \cite{hobolth2009, rao13a}. In broader settings parametrized by \textit{hazard} functions, \textit{dependent thinning} alternatives \cite[see][]{rao2012mcmc,miasojedow2017geometric} offer the only computationally feasible approach. Overall, data augmentation procedures in all the above instances are rigid, designed with small MJP systems in mind and only accommodate restrictive observation models suitable to few applications. Importantly, they often do not work (or do not scale) for the analysis of population models, where transition rates scale quadratically through interactions of marginal counts, observations are often a consequence of system jumps and unbounded populations are the norm. 

\noindent\textbf{Recent developments.} Current alternatives sit on top of the aforementioned benchmark algorithms, and are limited to addressing considerations of state-space explosions for stationary systems. In order to preserve asymptotic exactness, without imposing artificial bounds on population levels, sequential particle procedures may be used to target sequences of states in $\hat{\boldsymbol{x}}$ \citep{miasojedow2015particle} (subject to particle degeneracy), or arbitrary random truncations imposed over explorable spaces of paths \citep{georgoulas2017unbiased} (requiring costly \textit{Metropolis-Hastings} (M-H) acceptance steps to overcome induced bias). Most recent advances towards efficient algorithmic constructions \citep[see][]{zhang2017efficient} involve updating parameters $\boldsymbol{\lambda}$ \textit{within} forward-backward procedures for $(\hat{\boldsymbol{t}},\hat{\boldsymbol{x}})$, which works reportedly well with small MJP systems.

\subsection{Summary of contributions} \label{contributions}

In this work, we present a novel auxiliary-variable framework leading to data-augmentation techniques for conditional population model trajectories $(\boldsymbol{t},\boldsymbol{x})$ in \eqref{targetdens}. This will yield to computationally tractable joint distributions across both target and auxiliary variables, and readily lead to \textit{Gibbs}-like procedures satisfying detailed balance \cite[see][]{HigdonAux}. Hence, we further construct new MCMC samplers adaptable to popular Bayesian inferential tasks; in Figure \ref{fig:contributions} we summarize efficiency results that compare these to existing benchmark methods, in application to birth-death processes (left), stochastic epidemics (centre) and predator-prey (right) dynamics. The lines represent ratios in effective sample sizes across \textit{unknown} model parameters, tested at several population capacities specified by the horizontal axis. In each case, ratios are measured against a suitably chosen benchmark (horizontal line at level 1), and include confidence intervals through repetition over several datasets. Coloured lines correspond to samplers introduced in this paper; dark lines represent existing state of the art alternatives. In all cases, we note significant gains in scalability and efficiency.
\begin{figure}[h!]
  \centering
   \includegraphics[width=\linewidth]{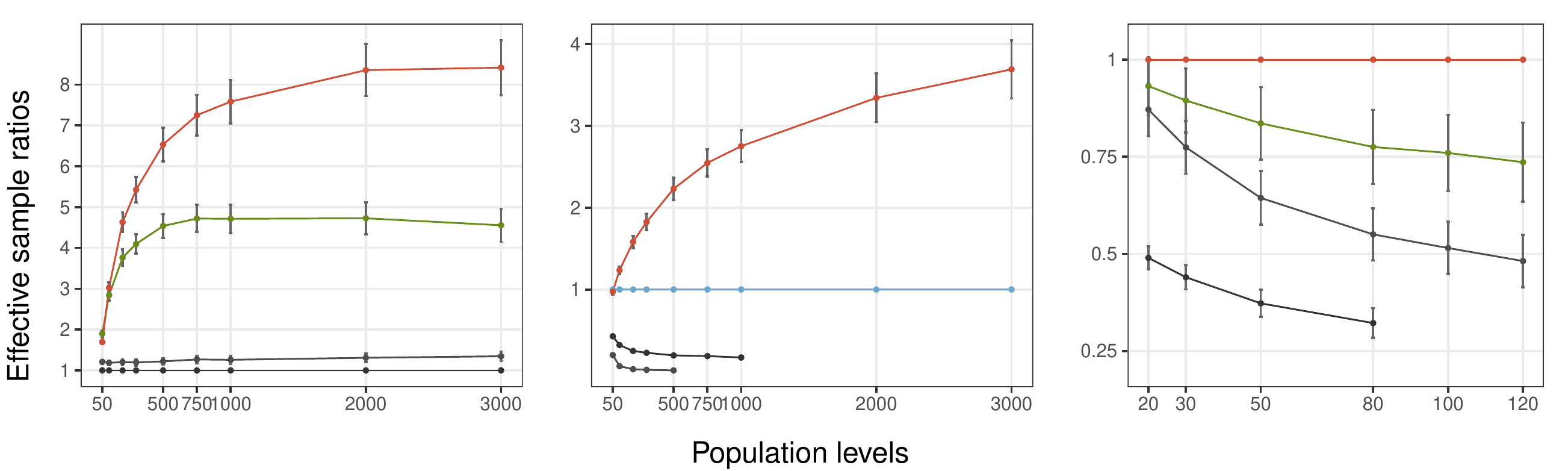}
  \caption{Ratios in effective sample sizes (with confidence intervals) across model parameters, for inferential tasks with birth-death (left), epidemic (centre) and predator-prey (right) systems. The horizontal axes represent population sizes tested. In each case, ratios are measured against a suitably chosen benchmark (horizontal line at level 1). Coloured lines correspond to techniques in this paper; dark lines represent state-of-the-art methods. We notice significant advances in system scalability (existing approaches are often unusable with large populations) and reasonable increments in efficiency in all cases.}
  \label{fig:contributions}
\end{figure}

Within the rest of the paper, Section \ref{effSection} introduces a \textit{two-step} data-augmentation with random \textit{importance weightings}; and further describes (i) performance optimization with stationary MJPs and (ii) limiting properties to population systems with infinite capacities. Also, the section draws comparisons and discusses differences with existing \textit{uniformization}-based methods, and addresses inference with (i) deterministic/random observations of population states, and (ii) observations of population jumps. Section \ref{ODEsection} presents auxiliary-variable results to efficiently sample jump trajectories as deviations from deterministic mean-average population dynamics; thus addressing considerations of state-space explosions \textit{strictly} within Gibbs procedures. Finally, Section \ref{splitting} studies dividing augmentation procedures into smaller computationally tractable counterparts.

\section{Uniformization and auxiliary variables} \label{Preliminaries} 

Let $(\hat{\boldsymbol{t}},\hat{\boldsymbol{x}})$ define an \textit{augmented} jump trajectory over the finite time interval $[0,T]$, so that $\hat{x}_i\in\mathcal{S}$ for $i\geq 0$. Here, inter-arrival times in $\hat{\boldsymbol{t}}$ are exponentially distributed with a fixed rate $$\Omega \geq \max_{x\in\mathcal{S}} \sup_{t\in[0,T]}|Q_x(t)|,$$
and $\hat{\boldsymbol{x}}$ is a realization from a discrete-time non-homogeneous Markov chain, with initial state $x_0\in\mathcal{S}$ drawn from $\pi(\cdot)$, and transition probability matrices $P(\hat{t}_i)=I+Q(\hat{t}_i)/\Omega, \, \hat{t}_i\in\hat{\boldsymbol{t}}.$
\begin{proposition} \label{Uniformization}
The process $(\hat{\boldsymbol{t}},\hat{\boldsymbol{x}})$ describes an \textit{augmented} MJP (allowing for self-transitions) on $\mathcal{X}$, and it is equivalent to $X=(\boldsymbol{t},\boldsymbol{x})$ with intensity matrices $\boldsymbol{Q}=\{Q(t) \, : \, t\geq 0\}$ and density function \eqref{pathProbs}.
\end{proposition}
This is a well known result; it follows since
\begin{align*}
\mathbb{P}(X_{t+s}=x'|X_t=x,\boldsymbol{Q}) = \sum_{k=0}^\infty \frac{s^k\Omega^k}{k!}e^{-s\Omega}\int_{\mathcal{H}^k} \big[P(u_1) \times \dots \times P(u_k)\big]_{x,x'} \, \mathrm{d}H(u_1,\dots,u_k)
\end{align*}
for all $x,x'\in\mathcal{S}$ and $0<t<t+s<T$ offers a \textit{randomized} or \textit{uniformized} representation of transition probabilities across times and states in the original MJP. There, $\mathrm{d}H(u_1,\dots,u_k) = k!/s^k \cdot \mathrm{d}u_1\dots\mathrm{d}u_k$ denotes the density of a $k$-dimensional vector of order statistics on $\mathcal{H}^k=\{(u_1,\dots,u_k)\in[t,t+s]^k \, : \, u_1<u_2<\dots<u_k\}$, and a proof of equivalence may be found on e.g. \cite{van1992uniformization,van2018uniformization}. Commonly used for simulation in homogeneous systems, the computational procedure that constructs these augmented sets of times $\hat{\boldsymbol{t}}=\{\hat{t}_0,\dots,\hat{t}_{m}\}$ and states $\hat{\boldsymbol{x}}=\{\hat{x}_0,\dots,\hat{x}_{m}\}$ offers an efficient alternative to \textit{Gillespie's algorithm}\citep{gillespie1977exact}, and is commonly referred to as \textit{uniformization} \cite[cf.][]{jensen1953markoff}. Whenever $\hat{x}_i=\hat{x}_{i-1}$, we refer to a transition $i$ as a \textit{virtual jump}. For example, in Figure \ref{virtualJumpExpl} (left) we observe an augmented \textit{birth-death} trajectory; there, we spot virtual jumps at times $\hat{t}_2,\hat{t}_6,\hat{t}_7$ and $\hat{t}_9$, which are represented by white circles on the time axis. On the right hand side, we observe the equivalent trajectory after virtual times and states have been removed.
\begin{figure*}[ht]
\vskip 0.1in
\begin{center}
\resizebox{\textwidth}{!}{%
\begin{tikzpicture}
\draw[->,line width=0.10mm] (0.1cm,-0.9cm) -- ++ (7.2cm,0cm);
\draw (0.1,-0.95cm) -- ++(0cm,0.1cm);
\node at (7.55cm,-0.9cm) {\small $\boldsymbol{t}$};
\node at (-0.5cm,-0.4cm) {\footnotesize $1$};
\node at (-0.5cm,0.3cm) {\footnotesize $2$};
\node at (-0.5cm,1cm) {\footnotesize $3$};
\node at (-0.5cm,1.7cm) {\footnotesize $4$};
\draw[dashed,draw=black!20!white,line width=0.10mm] (0.1cm,-0.4cm) -- (7.2cm,-0.4cm);
\draw[dashed,draw=black!20!white,line width=0.10mm] (0.1cm,0.3cm) -- (7.2cm,0.3cm); 
\draw[dashed,draw=black!20!white,line width=0.10mm] (0.1cm,1cm) -- (7.2cm,1cm);
\draw[dashed,draw=black!20!white,line width=0.10mm] (0.1cm,1.7cm) -- (7.2cm,1.7cm);
\draw[-,draw=black!80!white,line width=0.25mm] (0.1cm,-0.4cm) -- (0.6cm,-0.4cm); 
\draw[-,draw=black!80!white,line width=0.25mm] (0.6cm,-0.4cm) -- (0.6cm,0.3cm); 
\draw[-,draw=black!80!white,line width=0.25mm] (0.6cm,0.3cm) -- (2.4cm,0.3cm); 
\draw[-,draw=black!80!white,line width=0.25mm] (2.4cm,0.3cm) -- (2.4cm,1cm);
\draw[-,draw=black!80!white,line width=0.25mm] (2.4cm,1cm) -- (3cm,1cm); 
\draw[-,draw=black!80!white,line width=0.25mm] (3cm,1cm) -- (3cm,1.7cm); 
\draw[-,draw=black!80!white,line width=0.25mm] (3cm,1.7cm) -- (3.5cm,1.7cm);
\draw[-,draw=black!80!white,line width=0.25mm] (3.5cm,1.7cm) -- (3.5cm,1cm);
\draw[-,draw=black!80!white,line width=0.25mm] (3.5cm,1cm) -- (5.3cm,1cm);
\draw[-,draw=black!80!white,line width=0.25mm] (5.3cm,1cm) -- (5.3cm,0.3cm);
\draw[-,draw=black!80!white,line width=0.25mm] (5.3cm,0.3cm) -- (7.2cm,0.3cm);
\filldraw[black!80!white] (0.1cm,-0.4cm) circle [radius=0.04cm]; 
\filldraw[black!80!white] (0.6cm,0.3cm) circle [radius=0.04cm]; 
\filldraw[fill=black!00!white,draw=black!80!white] (1.6cm,0.3cm) circle [radius=0.05cm]; 
\filldraw[black!80!white] (2.4cm,1cm) circle [radius=0.04cm]; 
\filldraw[black!80!white] (3cm,1.7cm) circle [radius=0.04cm]; 
\filldraw[black!80!white] (3.5cm,1cm) circle [radius=0.04cm]; 
\filldraw[fill=black!00!white,draw=black!80!white] (4.3cm,1cm) circle [radius=0.05cm]; 
\filldraw[fill=black!00!white,draw=black!80!white] (4.8cm,1cm) circle [radius=0.05cm]; 
\filldraw[black!80!white] (5.3cm,0.3cm) circle [radius=0.04cm]; 
\filldraw[fill=black!00!white,draw=black!80!white] (6.5cm,0.3cm) circle [radius=0.05cm]; 
\node at (0.1cm,-0.1cm) {$\hat{x}_0$};
\node at (0.6cm,0.6cm) {$\hat{x}_1$};
\node at (1.6cm,0.6cm) {$\hat{x}_2$};
\node at (2.4cm,1.3cm) {$\hat{x}_3$};
\node at (3cm,2cm) {$\hat{x}_4$};
\node at (3.5cm,0.7cm) {$\hat{x}_5$};
\node at (4.3cm,1.3cm) {$\hat{x}_6$};
\node at (4.8cm,1.3cm) {$\hat{x}_7$};
\node at (5.3cm,0cm) {$\hat{x}_8$};
\node at (6.5cm,0.6cm) {$\hat{x}_9$};
\filldraw[fill=black!20!white,draw=black!80!white] (0.6cm,-0.9cm) circle [radius=0.06cm]; 
\filldraw[fill=black!00!white,draw=black!80!white] (1.6cm,-0.9cm) circle [radius=0.06cm]; 
\filldraw[fill=black!20!white,draw=black!80!white] (2.4cm,-0.9cm) circle [radius=0.06cm]; 
\filldraw[fill=black!20!white,draw=black!80!white] (3cm,-0.9cm) circle [radius=0.06cm]; 
\filldraw[fill=black!20!white,draw=black!80!white] (3.5cm,-0.9cm) circle [radius=0.06cm]; 
\filldraw[fill=black!00!white,draw=black!80!white] (4.3cm,-0.9cm) circle [radius=0.06cm]; 
\filldraw[fill=black!00!white,draw=black!80!white] (4.8cm,-0.9cm) circle [radius=0.06cm]; 
\filldraw[fill=black!20!white,draw=black!80!white] (5.3cm,-0.9cm) circle [radius=0.06cm]; 
\filldraw[fill=black!00!white,draw=black!80!white] (6.5cm,-0.9cm) circle [radius=0.06cm]; 
\node at (0.6cm,-1.4cm) {\small $\hat{t}_1$};
\node at (1.6cm,-1.4cm) {\small $\hat{t}_2$};
\node at (2.4cm,-1.4cm) {\small $\hat{t}_3$};
\node at (3cm,-1.4cm) {\small $\hat{t}_4$};
\node at (3.5cm,-1.4cm) {\small $\hat{t}_5$};
\node at (4.3cm,-1.4cm) {\small $\hat{t}_6$};
\node at (4.8cm,-1.4cm) {\small $\hat{t}_7$};
\node at (5.3cm,-1.4cm) {\small $\hat{t}_8$};
\node at (6.5cm,-1.4cm) {\small $\hat{t}_9$};

\draw[->,line width=0.10mm] (9cm,-0.9cm) -- ++ (7.2cm,0cm);
\draw (9,-0.95cm) -- ++(0cm,0.1cm);
\node at (16.55cm,-0.9cm) {\small $\boldsymbol{t}$};
\node at (8.4cm,-0.4cm) {\footnotesize $1$};
\node at (8.4cm,0.3cm) {\footnotesize $2$};
\node at (8.4cm,1cm) {\footnotesize $3$};
\node at (8.4cm,1.7cm) {\footnotesize $4$};
\draw[dashed,draw=black!20!white,line width=0.10mm] (9cm,-0.4cm) -- ++ (7.1cm,0cm);
\draw[dashed,draw=black!20!white,line width=0.10mm] (9cm,0.3cm) -- ++ (7.1cm,0cm); 
\draw[dashed,draw=black!20!white,line width=0.10mm] (9cm,1cm) -- ++ (7.1cm,0cm);
\draw[dashed,draw=black!20!white,line width=0.10mm] (9cm,1.7cm) -- ++ (7.1cm,0cm);
\draw[-,draw=black!80!white,line width=0.25mm] (9cm,-0.4cm) -- (9.5cm,-0.4cm); 
\draw[-,draw=black!80!white,line width=0.25mm] (9.5cm,-0.4cm) -- (9.5cm,0.3cm); 
\draw[-,draw=black!80!white,line width=0.25mm] (9.5cm,0.3cm) -- (11.3cm,0.3cm); 
\draw[-,draw=black!80!white,line width=0.25mm] (11.3cm,0.3cm) -- (11.3cm,1cm); 
\draw[-,draw=black!80!white,line width=0.25mm] (11.3cm,1cm) -- (11.9cm,1cm); 
\draw[-,draw=black!80!white,line width=0.25mm] (11.9cm,1cm) -- (11.9cm,1.7cm); 
\draw[-,draw=black!80!white,line width=0.25mm] (11.9cm,1.7cm) -- (12.4cm,1.7cm); 
\draw[-,draw=black!80!white,line width=0.25mm] (12.4cm,1.7cm) -- (12.4cm,1cm); 
\draw[-,draw=black!80!white,line width=0.25mm] (12.4cm,1cm) -- (14.2cm,1cm); 
\draw[-,draw=black!80!white,line width=0.25mm] (14.2cm,1cm) -- (14.2cm,0.3cm); 
\draw[-,draw=black!80!white,line width=0.25mm] (14.2cm,0.3cm) -- (16.1cm,0.3cm);
\filldraw[black!80!white] (9cm,-0.4cm) circle [radius=0.04cm]; 
\filldraw[black!80!white] (9.5cm,0.3cm) circle [radius=0.04cm]; 
\filldraw[black!80!white] (11.3cm,1cm) circle [radius=0.04cm]; 
\filldraw[black!80!white] (11.9cm,1.7cm) circle [radius=0.04cm]; 
\filldraw[black!80!white] (12.4cm,1cm) circle [radius=0.04cm]; 
\filldraw[black!80!white] (14.2cm,0.3cm) circle [radius=0.04cm]; 
\node at (9cm,-0.1cm) {$x_0$};
\node at (9.5cm,0.6cm) {$x_1$};
\node at (11.3cm,1.3cm) {$x_2$};
\node at (11.9cm,2cm) {$x_3$};
\node at (12.4cm,0.7cm) {$x_4$};
\node at (14.2cm,0cm) {$x_5$};
\filldraw[fill=black!20!white,draw=black!80!white] (9.5cm,-0.9cm) circle [radius=0.06cm]; 
\filldraw[fill=black!20!white,draw=black!80!white] (11.3cm,-0.9cm) circle [radius=0.06cm]; 
\filldraw[fill=black!20!white,draw=black!80!white] (11.9cm,-0.9cm) circle [radius=0.06cm]; 
\filldraw[fill=black!20!white,draw=black!80!white] (12.4cm,-0.9cm) circle [radius=0.06cm]; 
\filldraw[fill=black!20!white,draw=black!80!white] (14.2cm,-0.9cm) circle [radius=0.06cm]; 
\node at (9.5cm,-1.4cm) {\small $t_1$};
\node at (11.3cm,-1.4cm) {\small $t_2$};
\node at (11.9cm,-1.4cm) {\small $t_3$};
\node at (12.4cm,-1.4cm) {\small $t_4$};
\node at (14.2cm,-1.4cm) {\small $t_5$};

\end{tikzpicture}
}
\vskip 0in
\caption{Left, an augmented birth-death trajectory. Right, equivalent trajectory without virtual jumps.} 
\label{virtualJumpExpl}
\end{center}
\vskip -0.1in
\end{figure*}
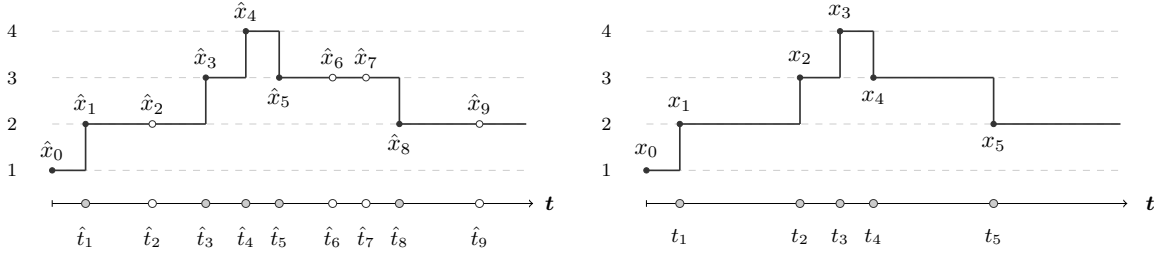

Along with the \textit{uniformized} trajectory $(\hat{\boldsymbol{t}},\hat{\boldsymbol{x}})$, let $\boldsymbol{u}=\{u_i\}_{i=1,\dots,m}$ define an auxiliary family of random variables, s.t.
\begin{align}
\mathbb{P}(u_i^{-1}(A)|\hat{\boldsymbol{x}}) = \int_{A}    g(a|u_{i-1},\hat{x}_{i-1},\hat{x}_i,\hat{t}_i)      \mu_{\mathcal{J}}(da) \label{intAux}
\end{align}
for all $i=2,\dots,m$ and $A\in\Sigma_{\mathcal{J}}$, with $\mathbb{P}(u_1^{-1}(A)|\hat{\boldsymbol{x}}) = \int_{A}    g(a|\hat{x}_0,\hat{x}_1,\hat{t}_1)      \mu_{\mathcal{J}}(da) $. Here, $(\mathcal{J},\Sigma_{\mathcal{J}})$ denotes an arbitrary support space and $\mu_{\mathcal{J}}$ is its corresponding base measure. We impose that a density/mass $g(\cdot)$ in \eqref{intAux} must be defined s.t. for any sequence $\boldsymbol{u}$, along with corresponding holding times in $\hat{\boldsymbol{t}}$, there must exist multiple probabilistically \textit{compatible} choices of $\hat{\boldsymbol{x}}$.

\begin{definition} 
Let $\boldsymbol{u}=\{u_i\}_{i=1,\dots,m}$ with $u_i\in\mathcal{J}$ be a sequence of auxiliary observations at times $\hat{\boldsymbol{t}}$ with $0\leq \hat{t}_1,\dots,\hat{t}_{m}\leq T$. We refer to a uniformized sequence $\hat{\boldsymbol{x}}$ as `compatible' with $\boldsymbol{u}$ given $\hat{\boldsymbol{t}}$ whenever $|\hat{\boldsymbol{x}}|=|\hat{\boldsymbol{t}}|$, $g(u_1|\hat{x}_0,\hat{x}_1,\hat{t}_1)>0$  and $g(u_i|u_{i-1},\hat{x}_{i-1},\hat{x}_i,\hat{t}_i)>0$ for all $i=2,\dots,m$.
\end{definition}

Trivially, a pair $(\hat{\boldsymbol{t}},\hat{\boldsymbol{x}})$ is compatible with $\boldsymbol{u}$ if a strictly positive mass is assigned by means of $g(\cdot|\hat{\boldsymbol{t}},\hat{\boldsymbol{x}})$ to the auxiliary realization. Conditioned on $\boldsymbol{u}$, we may restrict or assign importance weights across \textit{uniformized} trajectories in $\mathcal{X}$ within resampling procedures. Throughout the paper, the reader will be presented with multiple designs of densities $g$ in \eqref{intAux}, targeted both at general-form population models or specific jump systems in common application domains.

\noindent \textbf{Augmenting a trajectory through uniformization.} Assume the existence of a fixed, parametrized sequence of matrices $\boldsymbol{Q}=\boldsymbol{Q}(\boldsymbol{\lambda})$, a dominating rate $\Omega$ and an MJP trajectory $(\boldsymbol{t},\boldsymbol{x})\in\mathcal{X}$, s.t. $\boldsymbol{t}=\{t_0,\dots,t_{n}\}$ and $\boldsymbol{x}=\{x_0,\dots,x_{n}\}$. Then, we may sample an augmented pair $(\hat{\boldsymbol{t}},\hat{\boldsymbol{x}})$ from within the family of \textit{uniformized} representations equivalent to $(\boldsymbol{t},\boldsymbol{x})$. Marginalised over $\boldsymbol{u}$, a conditional density for times $\hat{\boldsymbol{t}}$ is, up to proportionality, given by
\begin{align}
f_{\hat{\boldsymbol{t}}}(\hat{t}_0,\dots,\hat{t}_m|\boldsymbol{t}, \boldsymbol{x},\Omega,\boldsymbol{Q}) 
& \propto \prod_{i=1}^n P_{x_{i-1},x_i}(t_i) \, \cdot \, \prod_{i=0}^n \prod_{j=0}^m P_{x_i,x_i}(\hat{t}_j)^{\mathbb{I}[\hat{t}_j \in (t_i,t_{i+1})]}      \, \cdot \, \Omega^m \, e^{-T\cdot\Omega} \nonumber \\
& \propto \prod_{i=0}^n \prod_{j=0}^m \big(\Omega + Q_{x_i}(\hat{t}_j)\big)^{\mathbb{I}[\hat{t}_j \in (t_i,t_{i+1})]}   \label{virtualJumps}
\end{align}
with $P(t)=I+Q(t)/\Omega$, $t_{n+1}=T$ and whenever $m\geq n$ and $\boldsymbol{t}\in \{\hat{t}_0,\dots,\hat{t}_m\}$. This corresponds to adding \textit{virtual} (self-transition) times to the sequence $\boldsymbol{t}$, by using successive Poisson processes with rates $\Omega + Q_{x_i}(t)$, $t\in(t_i,t_{i+1})$, for every $i=0,\dots,n$ \cite[cf.][]{rao13a}. 

Next, a \textit{uniformized} sequence of states $\hat{\boldsymbol{x}}$ can be deterministically assigned given knowledge of $\hat{\boldsymbol{t}},\boldsymbol{t},\boldsymbol{x}$, and an auxiliary sequence $\boldsymbol{u}|\hat{\boldsymbol{x}}$ sampled from a mass/density $g(\cdot)$ in \eqref{intAux}. These steps correspond to the top left/right diagrams in the birth-death example within Figure \ref{samplingExample}. There, a trajectory $(\boldsymbol{t},\boldsymbol{x})$ is complemented with virtual jumps (white circles on horizontal axis), states (white circles within trajectory) and auxiliary evidence (blue rectangles on some virtual epochs). In this example, $\boldsymbol{u}$ represents randomly \textit{locked} or \textit{clammed} jumps, and will become clearer to the reader soon. 
\begin{figure*}[h!]
\vskip 0.1in
\begin{center}
\resizebox{0.49\textwidth}{!}{%
\begin{tikzpicture}
\draw[->,line width=0.10mm] (0.1cm,-1.6cm) -- ++ (9.6,0cm);
\draw (0.1,-1.65cm) -- ++(0cm,0.1cm);
\foreach \i in {0,...,5} {
 \node at (-0.5,-1.1 + 0.5*\i) {\small $\i$}; \draw[dashed,draw=black!10!white,line width=0.10mm] (0.1,-1.1 + 0.5*\i) -- (9.6,-1.1 + 0.5*\i);
}

\draw[-,draw=black!80!white,line width=0.25mm] (0.1,-1.1 + 0.5*1) -- ++ (0.7,0); \filldraw[fill=black!80!white,draw=black!80!white] (0.1,-1.1 + 0.5*1) circle [radius=0.06cm]; 
\draw[-,draw=black!80!white,line width=0.25mm] (0.8,-0.6) -- (0.8,-0.1);

\draw[-,draw=black!80!white,line width=0.25mm] (0.8,-1.1 + 0.5*2) -- ++ (1.3,0); \filldraw[fill=black!80!white,draw=black!80!white] (0.8,-1.1 + 0.5*2) circle [radius=0.06cm]; 

\draw[-,draw=black!80!white,line width=0.25mm] (2.1,-0.1) -- ++ (1.1,0);
\draw[-,draw=black!80!white,line width=0.25mm] (3.2,-0.1) -- (3.2,0.4);

\draw[-,draw=black!80!white,line width=0.25mm] (3.2,0.4) -- ++ (0.8,0); \filldraw[fill=black!80!white,draw=black!80!white] (3.2,0.4) circle [radius=0.06cm]; 
\draw[-,draw=black!80!white,line width=0.25mm] (4,0.4) -- (4,0.9);

\draw[-,draw=black!80!white,line width=0.25mm] (4,-1.1 + 0.5*4) -- ++ (0.65,0); \filldraw[fill=black!80!white,draw=black!80!white] (4,-1.1 + 0.5*4) circle [radius=0.06cm]; 
\draw[-,draw=black!80!white,line width=0.25mm] (4.65cm,0.4cm) -- (4.65cm,0.9cm);

\draw[-,draw=black!80!white,line width=0.25mm] (4.65cm,0.4cm) -- ++ (1.05cm,0cm); \filldraw[fill=black!80!white,draw=black!80!white] (4.65cm,0.4cm) circle [radius=0.06cm]; 

\draw[-,draw=black!80!white,line width=0.25mm] (5.7,-1.1 + 0.5*3) -- ++ (0.7,0); 

\draw[-,draw=black!80!white,line width=0.25mm] (6.4,-1.1 + 0.5*3) -- ++ (0.6,0); 
\draw[-,draw=black!80!white,line width=0.25mm] (7,-0.1cm) -- (7,0.4cm);

\draw[-,draw=black!80!white,line width=0.25mm] (7,-0.1) -- ++ (1.6,0); \filldraw[fill=black!80!white,draw=black!80!white] (7,-0.1) circle [radius=0.06cm]; 

 \draw[-,draw=black!80!white,line width=0.25mm] (8.6,-0.1) -- ++ (1,0); 

\foreach \i in {0.8,3.2,4,4.65,7} {
\filldraw[fill=black!20!white,draw=black!80!white] (\i,-1.6) circle [radius=0.07cm]; 
}
\node at (0.8,-2.1) {\small $t_1$};;
\node at (3.2,-2.1) {\small $t_2$};
\node at (4,-2.1) {\small $t_3$};
\node at (4.65,-2.1) {\small $t_4$};
\node at (7,-2.1) {\small $t_5$};
\node at (0.1cm,-0.25cm) {$x_0$};
\node at (0.8cm,0.25cm) {$x_1$};
\node at (3.2cm,0.75cm) {$x_2$};
\node at (4cm,1.25cm) {$x_3$};
\node at (4.65cm,0.05cm) {$x_4$};
\node at (7cm,-0.45cm) {$x_5$};
\end{tikzpicture}}\hfill
\resizebox{0.49\textwidth}{!}{%
\begin{tikzpicture}
\draw[->,line width=0.10mm] (0.1cm,-1.6cm) -- ++ (9.6,0cm);
\draw (0.1,-1.65cm) -- ++(0cm,0.1cm);
\foreach \i in {0,...,5} {
 \node at (-0.5,-1.1 + 0.5*\i) {\small $\i$}; \draw[dashed,draw=black!10!white,line width=0.10mm] (0.1,-1.1 + 0.5*\i) -- (9.6,-1.1 + 0.5*\i);
}
\fill[draw=black!80!white,fill=blue!20!white] (1.95,-0.25) rectangle (2.25,0.05);
\fill[draw=black!80!white,fill=blue!20!white] (8.45,-0.25) rectangle (8.75,0.05);

\draw[-,draw=black!80!white,line width=0.25mm] (0.1,-1.1 + 0.5*1) -- ++ (0.7,0); \filldraw[fill=black!80!white,draw=black!80!white] (0.1,-1.1 + 0.5*1) circle [radius=0.06cm]; 
\draw[-,draw=black!80!white,line width=0.25mm] (0.8,-0.6) -- (0.8,-0.1);

\draw[-,draw=black!80!white,line width=0.25mm] (0.8,-1.1 + 0.5*2) -- ++ (1.3,0); \filldraw[fill=black!80!white,draw=black!80!white] (0.8,-1.1 + 0.5*2) circle [radius=0.06cm]; 

\draw[-,draw=black!80!white,line width=0.3mm] (2.1,-0.1) -- ++ (1.1,0); \filldraw[fill=black!00!white,draw=black!80!white] (2.1,-0.1) circle [radius=0.07cm]; 
\draw[-,draw=black!80!white,line width=0.3mm] (3.2,-0.1) -- (3.2,0.4);

\draw[-,draw=black!80!white,line width=0.25mm] (3.2,0.4) -- ++ (0.8,0); \filldraw[fill=black!80!white,draw=black!80!white] (3.2,0.4) circle [radius=0.06cm]; 
\draw[-,draw=black!80!white,line width=0.25mm] (4,0.4) -- (4,0.9);

\draw[-,draw=black!80!white,line width=0.25mm] (4,-1.1 + 0.5*4) -- ++ (0.65,0); \filldraw[fill=black!80!white,draw=black!80!white] (4,-1.1 + 0.5*4) circle [radius=0.06cm]; 
\draw[-,draw=black!80!white,line width=0.25mm] (4.65cm,0.4cm) -- (4.65cm,0.9cm);

\draw[-,draw=black!80!white,line width=0.25mm] (4.65cm,0.4cm) -- ++ (1.05cm,0cm); \filldraw[fill=black!80!white,draw=black!80!white] (4.65cm,0.4cm) circle [radius=0.06cm]; 

\draw[-,draw=black!80!white,line width=0.25mm] (5.7,-1.1 + 0.5*3) -- ++ (0.7,0); \filldraw[fill=black!00!white,draw=black!80!white] (5.7,-1.1 + 0.5*3) circle [radius=0.07cm]; 

\draw[-,draw=black!80!white,line width=0.25mm] (6.4,-1.1 + 0.5*3) -- ++ (0.6,0); \filldraw[fill=black!00!white,draw=black!80!white] (6.4,-1.1 + 0.5*3) circle [radius=0.07cm]; 
\draw[-,draw=black!80!white,line width=0.25mm] (7,-0.1cm) -- (7,0.4cm);

\draw[-,draw=black!80!white,line width=0.3mm] (7,-0.1) -- ++ (1.6,0); \filldraw[fill=black!80!white,draw=black!80!white] (7,-0.1) circle [radius=0.06cm]; 

 \draw[-,draw=black!80!white,line width=0.25mm] (8.6,-0.1) -- ++ (1,0); \filldraw[fill=black!00!white,,draw=black!80!white] (8.6,-0.1) circle [radius=0.07cm]; 

\foreach \i in {2.1,5.7,6.4,8.6} {
\filldraw[fill=black!00!white,draw=black!80!white] (\i,-1.6) circle [radius=0.07cm]; 
}
\foreach \i in {0.8,3.2,4,4.65,7} {
\filldraw[fill=black!20!white,draw=black!80!white] (\i,-1.6) circle [radius=0.07cm]; 
}
\node at (0.8,-2.06) {\small $\hat{t}_1$};
\node at (2.1,-2.06) {\small $\hat{t}_2$};
\node at (3.2,-2.06) {\small $\hat{t}_3$};
\node at (4,-2.06) {\small $\hat{t}_4$};
\node at (4.65,-2.06) {\small $\hat{t}_5$};
\node at (5.7,-2.06) {\small $\hat{t}_6$};
\node at (6.4,-2.06) {\small $\hat{t}_7$};
\node at (7,-2.06) {\small $\hat{t}_8$};
\node at (8.6,-2.06) {\small $\hat{t}_9$};
\node at (0.1cm,-0.25cm) {$\hat{x}_0$};
\node at (0.8cm,0.25cm) {$\hat{x}_1$};
\node at (2.1cm,0.25cm) {$\hat{x}_2$};
\node at (3.2cm,0.75cm) {$\hat{x}_3$};
\node at (4cm,1.25cm) {$\hat{x}_4$};
\node at (4.65cm,0.05cm) {$\hat{x}_5$};
\node at (5.7cm,0.75cm) {$\hat{x}_6$};
\node at (6.4cm,0.75cm) {$\hat{x}_7$};
\node at (7cm,-0.45cm) {$\hat{x}_8$};
\node at (8.6cm,0.25cm) {$\hat{x}_9$};
\end{tikzpicture}}
\vspace{5pt}

\resizebox{0.49\textwidth}{!}{%
\begin{tikzpicture}
\draw[->,line width=0.10mm] (0.1cm,-1.6cm) -- ++ (9.6,0cm);
\draw (0.1,-1.65cm) -- ++(0cm,0.1cm);
\foreach \i in {0,...,5} {
 \node at (-0.5,-1.1 + 0.5*\i) {\small $\i$}; \draw[dashed,draw=black!10!white,line width=0.10mm] (0.1,-1.1 + 0.5*\i) -- (9.6,-1.1 + 0.5*\i);
}

\draw[dashed,draw=black!20!white,line width=0.10mm] (2.1,-1.6) -- (2.1,1.5);
\draw[dashed,draw=black!20!white,line width=0.10mm] (8.6,-1.6) -- (8.6,1.5);
\node at (2.1,-1.6) {\color{blue!50!gray}\ding{55}};
\node at (8.6,-1.6) {\color{blue!50!gray}\ding{55}};

\foreach \i in {0.1,0.8,3.2,4,4.65,5.7,6.4,7} {
\filldraw[fill=black!00!white,draw=black!80!white] (\i,-1.6) circle [radius=0.07cm]; 
}
\node at (0.8,-2.06) {\small $\hat{t}_1$};
\node at (3.2,-2.06) {\small $\hat{t}_2$};
\node at (4,-2.06) {\small $\hat{t}_3$};
\node at (4.65,-2.06) {\small $\hat{t}_4$};
\node at (5.7,-2.06) {\small $\hat{t}_5$};
\node at (6.4,-2.06) {\small $\hat{t}_6$};
\node at (7,-2.06) {\small $\hat{t}_7$};
\end{tikzpicture}}\hfill\resizebox{0.49\textwidth}{!}{%
\begin{tikzpicture}
\draw[->,line width=0.10mm] (0.1cm,-1.6cm) -- ++ (9.6,0cm);
\draw (0.1,-1.65cm) -- ++(0cm,0.1cm);
\foreach \i in {0,...,5} {
 \node at (-0.5,-1.1 + 0.5*\i) {\small $\i$}; \draw[dashed,draw=black!10!white,line width=0.10mm] (0.1,-1.1 + 0.5*\i) -- (9.6,-1.1 + 0.5*\i);
}

\draw[-,draw=black!80!white,line width=0.25mm] (0.1,-1.1 + 0.5*3) -- ++ (0.7,0); \filldraw[fill=black!80!white,draw=black!80!white] (0.1,-1.1 + 0.5*3) circle [radius=0.06cm]; 
\draw[-,draw=black!80!white,line width=0.25mm] (0.8,0.4) -- (0.8,-0.1);

\draw[-,draw=black!80!white,line width=0.25mm] (0.8,-1.1 + 0.5*2) -- ++ (1.3,0); \filldraw[fill=black!80!white,draw=black!80!white] (0.8,-1.1 + 0.5*2) circle [radius=0.06cm]; 

\draw[-,draw=black!80!white,line width=0.3mm] (2.1,-0.1) -- ++ (1.1,0); 
\draw[-,draw=black!80!white,line width=0.3mm] (3.2,-0.1) -- (3.2,0.4);

\draw[-,draw=black!80!white,line width=0.25mm] (3.2,0.4) -- ++ (2.5,0); \filldraw[fill=black!80!white,draw=black!80!white] (3.2,0.4) circle [radius=0.06cm]; 
\draw[-,draw=black!80!white,line width=0.3mm] (5.7,-0.1) -- (5.7,0.4);

\draw[-,draw=black!80!white,line width=0.25mm] (5.7,-1.1 + 0.5*2) -- ++ (0.7,0); \filldraw[fill=black!80!white,draw=black!80!white] (5.7,-1.1 + 0.5*2) circle [radius=0.06cm]; 
\draw[-,draw=black!80!white,line width=0.3mm] (6.4,-0.1) -- (6.4,-0.6);

\draw[-,draw=black!80!white,line width=0.25mm] (6.4,-1.1 + 0.5*1) -- ++ (0.6,0); \filldraw[fill=black!80!white,draw=black!80!white] (6.4,-1.1 + 0.5*1) circle [radius=0.06cm]; 
\draw[-,draw=black!80!white,line width=0.25mm] (7,-0.1cm) -- (7,-0.6cm);

\draw[-,draw=black!80!white,line width=0.3mm] (7,-0.1) -- ++ (1.6,0); \filldraw[fill=black!80!white,draw=black!80!white] (7,-0.1) circle [radius=0.06cm]; 

 \draw[-,draw=black!80!white,line width=0.25mm] (8.6,-0.1) -- ++ (1,0); 

\foreach \i in {0.8,3.2,5.7,6.4,7} {
\filldraw[fill=black!20!white,draw=black!80!white] (\i,-1.6) circle [radius=0.07cm]; 
}
\node at (0.8,-2.1) {\small $t_1$};
\node at (3.2,-2.1) {\small $t_2$};
\node at (5.7,-2.1) {\small $t_3$};
\node at (6.4,-2.1) {\small $t_4$};
\node at (7,-2.1) {\small $t_5$};
\node at (0.1cm,0.75cm) {$x_0$};
\node at (0.8cm,-0.45cm) {$x_1$};
\node at (3.2cm,0.75cm) {$x_2$};
\node at (5.7cm,-0.45cm) {$x_3$};
\node at (6.4cm,-0.95cm) {$x_4$};
\node at (7cm,0.25cm) {$x_5$};
\end{tikzpicture}}
\vskip 0in
\caption{Sketch of a \textit{birth-death} sampling iteration with an auxiliary-variable \textit{naive} procedure. Top left, a reference path $(\boldsymbol{t},\boldsymbol{x})$; top right, augmentation with virtual jumps and states (represented by white circles), and auxiliary evidence $\boldsymbol{u}$ (blue rectangles). Bottom left, fixed Poisson holding times $\hat{\boldsymbol{t}}$ for a new trajectory; instances with rectangles are removed in this example. Bottom right, new trajectory sampled from within a compatible space in $\mathcal{X}$, with a forward-backward procedure conditioned on $\boldsymbol{u}$; virtual epochs have been removed.} 
\label{samplingExample}
\end{center}
\vskip -0.1in
\end{figure*}
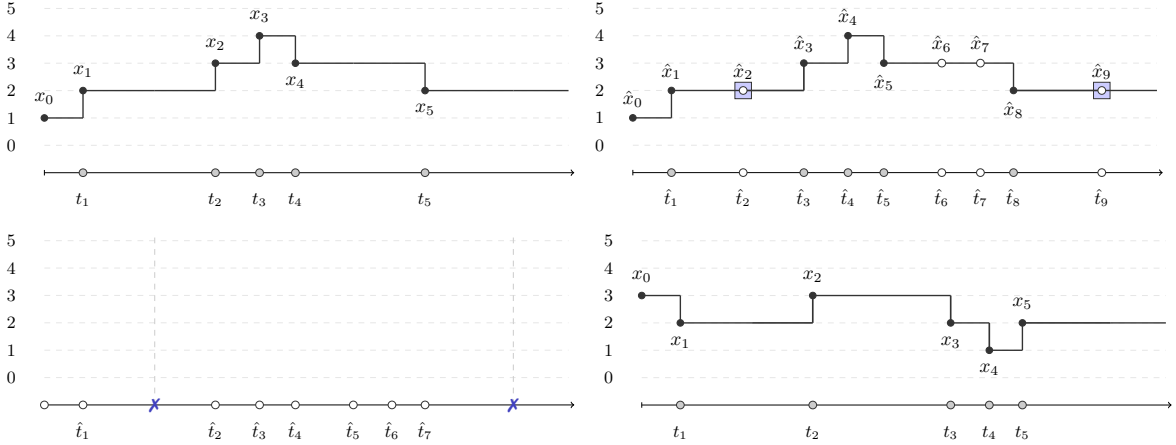

\noindent \textbf{Resampling a new trajectory according to compatibility rules.} A new trajectory within a restricted space of $\mathcal{X}$ may be obtained, by sampling a fresh augmented sequence $\hat{\boldsymbol{x}}|\hat{\boldsymbol{t}},\boldsymbol{u},\boldsymbol{Q}(\boldsymbol{\lambda})$ and removing all virtual entries. To this end, we ought to target the discrete-time representation
\begin{align}
f_{\hat{\boldsymbol{x}}}(\hat{x}_0,\dots,\hat{x}_{m}|\hat{\boldsymbol{t}},\boldsymbol{u},\boldsymbol{\lambda}) & \propto   f_{\hat{\boldsymbol{x}}}(\hat{x}_0,\dots,\hat{x}_{m}|\hat{\boldsymbol{t}},Q(\boldsymbol{\lambda})) \cdot g(u_1|\hat{x}_{0},\hat{x}_1,\hat{t}_1) \cdot \prod_{i=2}^m g(u_i|u_{i-1},\hat{x}_{i-1},\hat{x}_i,\hat{t}_i) \nonumber \\
& \propto \pi(\hat{x}_0) \cdot g(u_1|\hat{x}_{0},\hat{x}_1,\hat{t}_1) \cdot \prod_{i=2}^m  g(u_i|u_{i-1},\hat{x}_{i-1},\hat{x}_i,\hat{t}_i) \cdot   \prod_{i=1}^m  P_{\hat{x}_{i-1},\hat{x}_i}(\hat{t}_i), \label{StSpRep}
\end{align}
which readily simplifies to forward-backward steps with initial distribution $\pi(x)$, no importance updates and (non-stochastic) \textit{transition weight} matrices $\tilde{P}(\hat{t}_i ; \boldsymbol{u})$, defined s.t. 
\begin{align}
\tilde{P}_{\hat{x}_{i-1},x}(\hat{t}_i;  \boldsymbol{u}) = \mathbb{P}(\hat{x}_i = x,u_i|\hat{x}_{i-1}, u_{i-1},\hat{t}_i) = g(u_i|u_{i-1},\hat{x}_{i-1},\hat{x}_i=x,\hat{t}_i) \cdot P_{\hat{x}_{i-1},x} (\hat{t}_i) \label{seqTransMat}
\end{align}
for all states $x\in\mathcal{S}$ and epochs $i=1,\dots,m$. This step corresponds to the bottom left/right diagrams in Figure \ref{samplingExample}. On the left, we see an \textit{empty} frame of Poisson holding times $\hat{\boldsymbol{t}}$, which defines a random time-discretization of the time interval $[0,T]$. In this example $\boldsymbol{u}$ is defined so that, whenever a blue rectangle is shown, $\tilde{P}_{x,x'}(\hat{t}_i ;\boldsymbol{u})=0$ for all $x\neq x'$ within $\mathcal{S}$; hence, these epochs correspond with self-transitions in any newly sampled sequence of states. On the right, we find a new trajectory after a forward-backward pass and discarding all virtual transitions. 

\section{Efficient augmentation over restricted sets of candidate times} \label{effSection}

Next, we present novel designs of auxiliary variables associated with \textit{reference} uniformization-based data-augmentation algorithms, and further highlight the shortcomings of traditional methods for inference with population models \cite[see e.g.][and references therein]{hobolth2009,rao13a}. In later sections, we build over these results in order to scale sampling procedures, leading to efficiency results reported in Subsection \ref{contributions}. 

\subsection{Two-step data augmentation} \label{twoSetpSubsec}

To begin with, let $\boldsymbol{u}=\{u_i\}_{i=1,\dots,m}$ in \eqref{intAux} be defined on some set $\mathcal{J} = \{\phi,\bar{\phi}\}$, where $\phi$ denotes an arbitrary undefined \textit{open} element, and $\bar{\phi}$ is a complementary \textit{locked} element. Throughout this section, elements of $\boldsymbol{u}$ are assumed mutually independent given $(\hat{\boldsymbol{t}},\hat{\boldsymbol{x}})$. We define a probability mass function for the conditional distribution $\phi|\hat{x}_{i-1},\hat{x}_i,\hat{t}_i$ with respect to a suitable count measure, for all $(\hat{t}_i,\hat{x}_{i-1},\hat{x}_i)\in[0,T]\times\mathcal{S}^2$ as follows:
\begin{align}
g(\phi|\hat{x}_{i-1},\hat{x}_i,\hat{t}_i) = \frac{\psi(\hat{t}_i,\hat{x}_i)}{\Omega+Q_{\hat{x}_i}(\hat{t}_i)} \quad \text{and} \quad  g(\bar{\phi}|\hat{x}_{i-1},\hat{x}_i,\hat{t}_i) = 1 - g(\phi|\hat{x}_{i-1},\hat{x}_i,\hat{t}_i),    \quad \text{if} \quad \hat{x}_{i-1}=\hat{x}_i,  \label{densGpoissonsplit}
\end{align}
with $g(\phi|\hat{x}_{i-1},\hat{x}_i,\hat{t}_i) = 1$ otherwise. Here, $\psi:[0,T]\times\mathcal{S}\rightarrow\mathbb{R}_+$ is any operator that assigns real-valued \textit{intensities} across time to the various states in $\mathcal{S}$, and must satisfy $\psi(t,x)\in(0,\Omega + Q_x(t)]$ for all $(t,x)\in[0,T]\times\mathcal{S}$.

\begin{proposition} \label{propDataAug}
Let $X=(\boldsymbol{t},\boldsymbol{x})$ be an MJP realization with \textit{intensity} $\boldsymbol{Q}=\{Q(t) \, : \, t\geq 0\}$, s.t. $\boldsymbol{t}=\{t_0,\dots,t_{n}\}$ and $\boldsymbol{x}=\{x_0,\dots,x_{n}\}$. Consider an augmentation procedure for $X$, to a triplet $(\hat{\boldsymbol{t}},\hat{\boldsymbol{x}},\boldsymbol{u})$, where 
\begin{itemize}
\item A sequence $\hat{\boldsymbol{t}}$ augments $\boldsymbol{t}$ by adding virtual times from two jointly independent Poisson processes, 
\begin{itemize}
\item a `controlled' process with rate $\psi(t,x_i)>0$, $t\in(t_i,t_{i+1})$ within intervals of $\boldsymbol{t}$, and
\item a `compensating' process with rate $\Omega + Q_{x_i}(t) - \psi(t,x_i)$, $t\in (t_i,t_{i+1}), i\geq 0$.
\end{itemize}
\item An augmented sequence of states $\hat{\boldsymbol{x}}$ is deterministically assigned given knowledge of $\hat{\boldsymbol{t}},\boldsymbol{t},\boldsymbol{x}$.
\item Auxiliary variables $\boldsymbol{u}$ are deterministically assigned so that
\begin{itemize}
\item $u_i = \bar{\phi}$ for all $i\geq 1$ where time $\hat{t}_i$ in $\hat{\boldsymbol{t}}$ was sampled from the `compensating' Poisson process,
\item $u_i = \phi$ otherwise; i.e. either $\hat{t}_i$ was sampled from the `controlled' process, or $\hat{t}_i\in\boldsymbol{t}$. 
\end{itemize}
\end{itemize}
Then, this construction yields an statistically equivalent triplet $(\hat{\boldsymbol{t}},\hat{\boldsymbol{x}},\boldsymbol{u})$, when compared to sampling $(\hat{\boldsymbol{t}},\hat{\boldsymbol{x}})$ from \eqref{virtualJumps} followed by auxiliary variables $\boldsymbol{u}$ from \eqref{densGpoissonsplit}.
\end{proposition}

Due to Markovian properties, we only require to test equivalence in density representations for realizations $(\hat{\boldsymbol{t}},\hat{\boldsymbol{x}},\boldsymbol{u})$ restricted to intervals $[t_i,t_{i+1}), i\geq 0$. This is however straightforward, by noting that $\bar{\phi}$ can only be sampled from the \textit{`compensating'} Poisson process in the newly described data augmentation procedure; thus, the proof is omitted.
At a basic level, we note that \textit{locked} elements $\bar{\phi}$ in $\boldsymbol{u}$ may only be a consequence of virtual transitions within any \textit{uniformized} trajectory representation $(\hat{\boldsymbol{t}},\hat{\boldsymbol{x}})\in\mathcal{X}$; in Figure \ref{samplingExample}, these correspond to the blue rectangles on the top-right diagram. 

In view of \eqref{densGpoissonsplit}, note that \textit{forward filtering} steps for a sequence $\hat{\boldsymbol{x}}$ in \eqref{StSpRep}, conditioned on $\boldsymbol{u}$, reduce to
\begin{align*}
\mathbb{P}(\hat{x}_i=x|u_1,\dots,u_{i-1},u_i=\bar{\phi};\hat{\boldsymbol{t}}) & = \sum_{x'\in\mathcal{S}} \mathbb{P}(\hat{x}_i=x,\hat{x}_{i-1}=x'|u_1,\dots,u_{i-1},u_i=\bar{\phi};\hat{\boldsymbol{t}}) \\
& \propto g(\bar{\phi}|x,x,\hat{t}_i) \cdot P_{x,x} (\hat{t}_i) \cdot \mathbb{P}(\hat{x}_{i-1}=x|u_1,\dots,u_{i-1};\hat{\boldsymbol{t}})  \\
& \propto \bigg(1+\frac{Q_x(\hat{t}_i)-\psi(\hat{t}_i,x)}{\Omega}\bigg) \cdot \mathbb{P}(\hat{x}_{i-1}=x|u_1,\dots,u_{i-1};\hat{\boldsymbol{t}}),
\end{align*}
for all $x\in\mathcal{S}$ and whenever $u_i=\bar{\phi}$, $i\geq 1$; i.e. there exists a direct probabilistic correspondence across states over \textit{locked} time epochs. Moreover, if $u_j = \bar{\phi}$ for subsequent $j = i+1,\dots,i+k$, it follows
\begin{align}
\mathbb{P}(\hat{x}_{i+k}=x|u_1,\dots,u_i,u_{i+1}=\bar{\phi},&\dots,u_{i+k}=\bar{\phi};\hat{\boldsymbol{t}}) \propto \nonumber \\ &\prod_{j=1}^k \bigg(1+\frac{Q_x(\hat{t}_{i+j})-\psi(\hat{t}_{i+j},x)}{\Omega}\bigg) \cdot  \mathbb{P}(\hat{x}_i=x|u_1,\dots,u_i;\hat{\boldsymbol{t}}), \label{weightEq}
\end{align}
for all $ x\in\mathcal{S}$. Similarly, probabilities for \textit{backward sampling} steps, conditioned on $\boldsymbol{u}$, are given by
\begin{align*}
\mathbb{P}(\hat{x}_{i} = x|\hat{x}_{i+1},u_{i+1}=\bar{\phi};\hat{\boldsymbol{t}}) =  \mathbb{I}(x=\hat{x}_{i+1}), \quad \text{for all} \, x\in\mathcal{S}, \, \text{and}\, i\geq 0\, \text{s.t.}\, u_{i+1}=\bar{\phi},
\end{align*}
and are thus deterministic at times with auxiliary \textit{locked} instances. Trivially, the output of a \textit{uniformized} augmented sequence $\hat{\boldsymbol{x}}$ must satisfy $\hat{x}_{i} = \hat{x}_{i+1} = \dots = \hat{x}_{i+k-1} = \hat{x}_{i+k}$ whenever $u_j = \bar{\phi}$ for all $j = i+1,\dots,i+k$. Hence, jumps in $\hat{\boldsymbol{x}}$ are restricted to times with auxiliary \textit{open} instances  $\phi$. Additionally, since these correspond to $\boldsymbol{t}$ along with random draws from the `\textit{controlled}' process in Proposition \ref{propDataAug}, a (non-stochastic) \textit{weight} matrix for transitions in $\hat{\boldsymbol{x}}$  follows from \eqref{seqTransMat}-\eqref{densGpoissonsplit}, i.e.
\begin{align}
\tilde{P}_{\hat{x}_{i-1},x}(\hat{t}_i; u_i=\phi) =  g(\phi|\hat{x}_{i-1},\hat{x}_i=x,\hat{t}_i) \cdot P_{\hat{x}_{i-1},x} (\hat{t}_i) =  \frac{Q_{\hat{x}_{i-1},x}(\hat{t}_i)}{\Omega} \label{probjumpOpen}
\end{align}
whenever $x\neq \hat{x}_{i-1}$, and 
\begin{align}
\tilde{P}_{\hat{x}_{i-1},x}(\hat{t}_i; u_i=\phi) =  g(\phi|\hat{x}_{i-1},\hat{x}_i=\hat{x}_{i-1},\hat{t}_i) \cdot P_{x,x} (\hat{t}_i) =  \frac{\psi(\hat{t}_i,\hat{x}_{i-1})}{\Omega} \label{probjumpVirtual}
\end{align}
otherwise. In conclusion, owing to \eqref{weightEq}-\eqref{probjumpVirtual} and Proposition \ref{propDataAug}, a sampler of sequentially correlated MJP trajectories $(\boldsymbol{t},\boldsymbol{x})|\boldsymbol{Q}(\boldsymbol{\lambda})$, where candidate times $\hat{\boldsymbol{t}}$ are governed by an arbitrary intensity $\psi(\cdot)$, is formalized in Algorithm \ref{naiveAlgo}. There, note that the dominating rate $\Omega$ within transition matrices $\tilde{P}(\cdot)$ in \eqref{transPdef} is dropped; this corresponds to equations \eqref{probjumpOpen}-\eqref{probjumpVirtual} and fades up to proportionality.

\begin{algorithm}[t]
  \normalsize
\caption{Naive construction of correlated MJP trajectories on $\mathcal{X}$.}
\label{naiveAlgo}
\vspace{2pt}
\begin{minipage}[t]{0.08\textwidth}
\textit{Input:}
\end{minipage}
\begin{minipage}[t]{0.9\textwidth}
Sequence of intensity matrices  $\boldsymbol{Q}=\boldsymbol{Q}(\boldsymbol{\lambda})$ parametrized by $\boldsymbol{\lambda}$. \\[1pt]
A (strictly) dominating rate $\Omega >\max_{x\in\mathcal{S}} \sup_{t\in[0,T]}|Q_x(t)|,$. \\[1pt]
An MJP trajectory $(\boldsymbol{t},\boldsymbol{x})\in\mathcal{X}$ with $\boldsymbol{t}=\{t_0,\dots,t_{n}\}$ and $\boldsymbol{x}=\{x_0,\dots,x_{n}\}$. \\[1pt]
Intensity operator $\psi:[0,T]\times\mathcal{S}\rightarrow\mathbb{R}_+$ for candidate times, s.t. $\psi(t,x)\leq\Omega + Q_x(t)$.
\end{minipage}
\vspace{4pt}

\begin{minipage}[t]{0.08\textwidth}
\textit{Output:}
\end{minipage}
\begin{minipage}[t]{0.9\textwidth}
A new MJP trajectory $(\boldsymbol{t},\boldsymbol{x})_{new}\in\mathcal{X}$ sampled from the density $f_{X}(\boldsymbol{t},\boldsymbol{x}|\boldsymbol{Q})$ in \eqref{pathProbs}.
\end{minipage}
\vspace{-6pt}

\hrulefill
\begin{algorithmic}[1] 
\State Create an (ordered) set of candidate times $\hat{\boldsymbol{t}}=\{\hat{t}_0,\dots,\hat{t}_{m}\}$, $m\geq n$, attaching to $\boldsymbol{t}$ auxiliary events from a Poisson process; rate $\psi(t,x_i)>0$, within intervals $(t_i,t_{i+1})$, with $t_{n+1}=T$. \vspace{2pt}
\State For $i=0,\dots,m$, draw random amount $k_i$ of \textit{weighting} times $\boldsymbol{s}_i=\{s_1,\dots,s_{k_i} \}$ over the interval $(\hat{t}_i,\hat{t}_{i+1})$; use compensating rate $\Omega + Q_{\hat{x}_i}(t) - \psi(t,\hat{x}_i)$, $t\in (\hat{t}_i,\hat{t}_{i+1})$. \vspace{2pt}
\State Draw a new state sequence $\hat{\boldsymbol{x}}=\{\hat{x}_0,\dots,\hat{x}_{m}\}$ with a forward-backward procedure; given initial distribution $\pi(x)$, transition \textit{weight} matrices
\begin{align}
\tilde{P}(\hat{t}_i) = \textnormal{diag}( \{\psi(\hat{t}_i,x)-Q_{x}(\hat{t}_i) \, : \, x\in\mathcal{S}\}) + Q(\hat{t}_i) , \label{transPdef}
\end{align}
and (random) importance weights
\begin{align}
w_i(x)=\prod_{s\in\boldsymbol{s}_i} \bigg(1+\frac{Q_x(s)-\psi(s,x)}{\Omega}\bigg) , \label{weightsDef}
\end{align} 
imposed over epochs $i\in\{0,\dots,m\}$.  \vspace{2pt}
\State Remove self-transitions on $(\hat{\boldsymbol{t}},\hat{\boldsymbol{x}})$ to produce $(\boldsymbol{t},\boldsymbol{x})_{new}$.
\end{algorithmic}
\end{algorithm}

We refer to Algorithm \ref{naiveAlgo} as a \textit{naive} approach; its purpose is to serve as a starting point. Noticeably, the procedure requires forward-backward steps. It is thus inefficient to sample plain MJP trajectories $X\equiv(\boldsymbol{t},\boldsymbol{x})\in\mathcal{X}$ subject to no observations, in comparison to a generative approach such as \textit{Gillespie's algorithm} for stationary systems \citep{gillespie1977exact}. However, Algorithm \ref{naiveAlgo} is readily amendable for conditioning on observations $\boldsymbol{O}=\{O_r\}_{r\geq 1}$ commonly encountered in applications. This is because, by assumption, observation models are independent of auxiliary jump events within augmented representations $(\hat{\boldsymbol{t}},\hat{\boldsymbol{x}})$, and further independent of auxiliary variables $\boldsymbol{u}$.

\noindent\textbf{Conditioning on system state observations.} In the traditional set-up, $\boldsymbol{O}=\{O_r\}_{r\geq 1}$ is a sequence of population level observations at (ordered) time points $t_r\in[0,T]$, $r\geq 1$, s.t. 
$\mathcal{L}(\boldsymbol{O}|X)=\prod_{r\geq 1}f(O_r|X_{t_r})$
for some mass/density function $f(\cdot)$ over an arbitrary support set. A conditional probability mass function for an augmented sequence of states $\hat{\boldsymbol{x}}$ is given by
\begin{align*}
f_{\hat{\boldsymbol{x}}}(\hat{x}_0,\dots,\hat{x}_{m}|\hat{\boldsymbol{t}},&\hat{\boldsymbol{s}}_1,\dots,\hat{\boldsymbol{s}}_m, \boldsymbol{Q}(\boldsymbol{\lambda}),\boldsymbol{O},\Omega) \propto  \\ 
& \pi(\hat{x}_0) \cdot \prod_{i=1}^m \tilde{P}_{\hat{x}_{i-1},\hat{x}_i}(\hat{t}_i) \cdot \prod_{i=0}^m \bigg[ \prod_{s\in\hat{\boldsymbol{s}}_i} \bigg(1+\frac{Q_{\hat{x}_i}(s)-\psi(s,\hat{x}_i)}{\Omega}\bigg) \cdot  \prod_{r : t_r\in[\hat{t}_i,\hat{t}_{i+1})}  f(O_r|\hat{x}_i) \bigg] , 
\end{align*}
by noting that a population observation $O_r$ at any time $t_r\in[t_i,t_{i+1}]$ is only a consequence of $\hat{x}_i$, for all $i=1,\dots,m$. Thus, an auxiliary variable sampling procedure as introduced in Algorithm \ref{naiveAlgo}, with importance weights in \eqref{weightsDef} replaced by
$$w_i(x)=\prod_{s\in\hat{\boldsymbol{s}}_i} \bigg(1+\frac{Q_{x}(s)-\psi(s,x)}{\Omega}\bigg) \cdot  \prod_{r : t_r\in[\hat{t}_i,\hat{t}_{i+1})}  f(O_r|x) , \quad x\in\mathcal{S}$$
for $i=0,\dots,m$, defines a Markov chain over MJP trajectories in $\mathcal{X}$, with stationary distribution $
f_{X}(\boldsymbol{t},\boldsymbol{x}|\boldsymbol{\lambda},\boldsymbol{O})$ in \eqref{targetdens}. Note that $\mathcal{L}(\boldsymbol{O}|X)$ can accommodate both random and deterministic observations (by means of identity functions); thus, this offers an adaptable exact framework not restricted to jump models subject to measurement error \citep[cf.][]{golightly2015bayesian}. Finally, if the population process is assumed stationary and $\psi(t,x) \equiv \psi(x) = \Omega + Q_{x}$, for all $(t,x)\in[0,T]\times\mathcal{S}$, then the full procedure simplifies to Algorithm 2 in \cite{rao13a}.

\noindent\textbf{Conditioning on system jump observations.} Relevant to epidemics, network queues and genetic chains, let  $\boldsymbol{O}=\{O_r\}_{r\geq 1}$ be a sequence of jump observations at time points $t_r\in[0,T]$, $r\geq 1$, s.t.
\begin{align}
\mathcal{L}(\boldsymbol{O}|X)= \prod_{i=1,\dots,n} \Big[ (1-p_{x_{i-1},x_i})  \cdot \prod_{r\geq 1} \mathbb{I}(t_i\neq t_r) + \sum_{r\geq 1} \mathbb{I}(t_i=t_r) \cdot p_{x_{i-1},x_i} \cdot f(O_r|x_{i-1},x_i) \Big] \label{likJumps}
\end{align} 
for trajectories $X=(\boldsymbol{t},\boldsymbol{x})$, where $p_{x,x'}\in[0,1]$, $x,x'\in\mathcal{S}$ denotes the probability that a process jump $x\rightarrow x'$ triggers an observation with a conditional density $f(\cdot)$; and $p_{x,x}=0$ for all $x\in\mathcal{S}$. Then, 
\begin{align*}
f_{\hat{\boldsymbol{x}}}(\hat{x}_0,\dots,\hat{x}_{m}|\hat{\boldsymbol{t}},&\hat{\boldsymbol{s}}_1,\dots,\hat{\boldsymbol{s}}_m, \boldsymbol{Q}(\boldsymbol{\lambda}),\boldsymbol{O},\Omega)  \propto\pi(\hat{x}_0) \cdot \prod_{i=0}^m  \prod_{s\in\hat{\boldsymbol{s}}_i} \bigg(1+\frac{Q_{\hat{x}_i}(s)-\psi(s,\hat{x}_i)}{\Omega}\bigg) \cdot \\
& \prod_{i=1}^m \tilde{P}_{\hat{x}_{i-1},\hat{x}_i}(\hat{t}_i)  \Big[ (1-p_{\hat{x}_{i-1},\hat{x}_i})  \cdot \prod_{r\geq 1} \mathbb{I}(\hat{t}_i\neq t_r) + \sum_{r\geq 1} \mathbb{I}(\hat{t}_i=t_r) \cdot p_{\hat{x}_{i-1},\hat{x}_i} \cdot f(O_r|\hat{x}_{i-1},\hat{x}_i) \Big] , 
\end{align*}
and a sampling procedure as introduced in Algorithm \ref{naiveAlgo}, where $\hat{P}$ in \eqref{transPdef} is replaced by a sequence of matrices $P_i$, $i=1,\dots,m$, s.t.
\begin{align*}
P_{x,x'}(\hat{t}_i) = Q_{x,x'}(\hat{t}_i)  \cdot  p_{x,x'} \cdot f(O_r|x,x') 
\end{align*}
whenever $\hat{t}_i=t_r$ for some $r\geq 1$, and
\begin{align*}
P_{x,x'}(\hat{t}_i) =  Q_{x,x'}(\hat{t}_i)  \cdot (1- p_{x,x'}) \quad \text{if} \; x\neq x', \quad \text{with} \quad P_{x,x}(\hat{t}_i) = \psi(\hat{t}_i,x),
\end{align*}
otherwise, defines a Markov chain over MJP trajectories in $\mathcal{X}$, with stationary distribution $
f_{X}(\boldsymbol{t},\boldsymbol{x}|\boldsymbol{\lambda},\boldsymbol{O})$ in \eqref{targetdens}. Above, equation \eqref{likJumps} is explained by the fact that, in common application areas, only certain types of jumps are \textit{observable}. For instance, removal times of infective individuals are often the basis for inferential epidemic studies, however, infectious times are never observed.

In all cases, the associated MCMC samplers yield ergodic Markov chains over posterior MJP trajectories. This is because, since matrices in $\boldsymbol{Q}$ are sparse, a conditional sequence $\hat{\boldsymbol{x}}|\hat{\boldsymbol{t}}$ is always supported within a finite product space of $n(\hat{\boldsymbol{t}})$ subsets of $\mathcal{S}$. However, $\psi(t,x)>0$ for all $(t,x)\in[0,T]\times\mathcal{S}$ by definition, and any full sequence in $\mathcal{X}$ is always accessible by sampling an appropriate number of transition times. All trajectories are thus aperiodic and positive recurrent; moreover, auxiliary variables leave the target marginal distribution unaltered and the sampler will reach the desired invariant distribution.

\subsection{Accelerating performance with stationary processes}

In the reduction to a strictly stationary system (so rates are independent of time), forward-filtering steps in \eqref{weightEq} reduce to
\begin{align*}
\mathbb{P}(\hat{x}_{i+k}=x|u_1,\dots,u_i,u_{i+1}=\bar{\phi},\dots,u_{i+k}=\bar{\phi};\hat{\boldsymbol{t}}) \propto \bigg[1+\frac{Q_x-\psi(x)}{\Omega}\bigg]^{k}  \cdot  \mathbb{P}(\hat{x}_i=x|u_1,\dots,u_i;\hat{\boldsymbol{t}}), 
\end{align*}
whenever $u_j=\bar{\phi}$, $j=i+1,\dots,i+k$. Hence, times in $\hat{\boldsymbol{t}}$ generated by a `\textit{compensating}' process in Proposition \ref{propDataAug} are of no relevance; only Poisson counts $k_i, \, i=0,\dots,m$ in Algorithm \ref{naiveAlgo} must be retained. Thus, a sampling scheme for stationary MJPs,  similarly adaptable to observations, reduces to Algorithm \ref{naiveAlgoUnif}.
\begin{algorithm}[t]
  \normalsize
\caption{Reduced two-step construction of correlated stationary MJP trajectories on $\mathcal{X}$.}
\label{naiveAlgoUnif}
\vspace{2pt}
\begin{minipage}[t]{0.08\textwidth}
\textit{Input:}
\end{minipage}
\begin{minipage}[t]{0.9\textwidth}
Infinitesimal generator matrix $Q=Q(\boldsymbol{\lambda})$ parametrized by $\boldsymbol{\lambda}$. \\[1pt]
A (strictly) dominating rate for $\Omega > \max_{x\in\mathcal{S}} |Q_x|$. \\[1pt]
An MJP trajectory $(\boldsymbol{t},\boldsymbol{x})\in\mathcal{X}$ with $\boldsymbol{t}=\{t_0,\dots,t_{n}\}$ and $\boldsymbol{x}=\{x_0,\dots,x_{n}\}$. \\[1pt]
\textit{Intensity} operator $\psi:\mathcal{S}\rightarrow\mathbb{R}_+$ for candidate times, s.t. $\psi(x)\leq\Omega + Q_x$ for all $x\in\mathcal{S}$.
\end{minipage}
\vspace{4pt}

\begin{minipage}[t]{0.08\textwidth}
\textit{Output:}
\end{minipage}
\begin{minipage}[t]{0.9\textwidth}
A new MJP trajectory $(\boldsymbol{t},\boldsymbol{x})_{new}\in\mathcal{X}$ sampled from the density $f_{X}(\boldsymbol{t},\boldsymbol{x}|Q)$ in \eqref{pathProbs}.
\end{minipage}
\vspace{-6pt}

\hrulefill
\begin{algorithmic}[1] 
\State Create an (ordered) set of candidate times $\hat{\boldsymbol{t}}=\{\hat{t}_0,\dots,\hat{t}_{m}\}$, $m\geq n$, attaching to $\boldsymbol{t}$ auxiliary events from a Poisson process; rate $\psi(x_i)>0$, within intervals $(t_i,t_{i+1})$, with $t_{n+1}=T$. \vspace{2pt}
\State Sample a sequence $\boldsymbol{k}=\{k_0,\dots,k_{m}\}$ of Poisson count variables with rates
\begin{align}
\big[\Omega + Q_{\hat{x}_{i}} - \psi({\hat{x}_{i}}) \big] \cdot (\hat{t}_{i+1}-\hat{t}_i), \quad i=0,\dots,m, \quad \text{s.t.} \quad \hat{t}_{m+1} = T. \label{countsK}
\end{align} 
\State Draw a new sequence $\hat{\boldsymbol{x}}=\{\hat{x}_0,\dots,\hat{x}_{m}\}$ with a forward-backward procedure; given initial distribution $\pi(x)$, transition \textit{weight} matrix
\begin{align*}
\tilde{P} = \textnormal{diag}( \{\psi(x)-Q_{x} \, : \, x\in\mathcal{S}\}) + Q , 
\end{align*}
and (random) importance weights
\begin{align*}
w_i(x)=\bigg(1+\frac{Q_x-\psi(x)}{\Omega}\bigg)^{k_i} , \quad i=0,\dots,m.
\end{align*} 
\State Remove self-transitions on $(\hat{\boldsymbol{t}},\hat{\boldsymbol{x}})$ to produce $(\boldsymbol{t},\boldsymbol{x})_{new}$.
\end{algorithmic}
\end{algorithm}
To aid the understanding of these results, Figure \ref{samplingExample2} shows an example with a graphical overview of a two-step data augmentation leading to count variables \eqref{countsK}. 
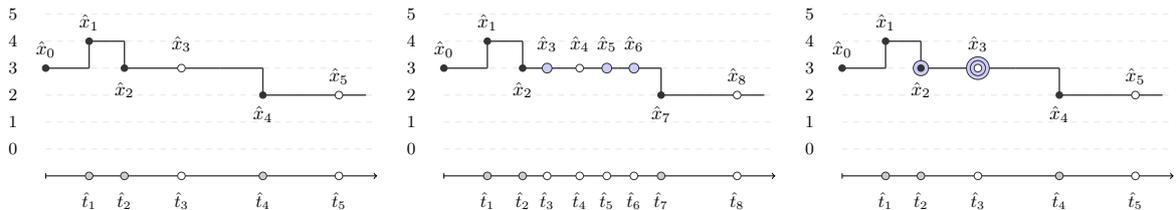
\begin{figure*}[b]
\begin{center}
\resizebox{0.32\textwidth}{!}{%
\begin{tikzpicture}
\draw[->,line width=0.10mm] (3.2cm,-1.6cm) -- ++ (6.1,0cm);
\draw (3.2,-1.65cm) -- ++(0cm,0.1cm);
\foreach \i in {0,...,5} {
 \node at (2.6,-1.1 + 0.5*\i) {\small $\i$}; \draw[dashed,draw=black!10!white,line width=0.10mm] (3.2,-1.1 + 0.5*\i) -- ++ (6.1,0);
}

\draw[-,draw=black!80!white,line width=0.25mm] (3.2,0.4) -- ++ (0.8,0); \filldraw[fill=black!80!white,draw=black!80!white] (3.2,0.4) circle [radius=0.06cm]; 
\draw[-,draw=black!80!white,line width=0.25mm] (4,0.4) -- (4,0.9);

\draw[-,draw=black!80!white,line width=0.25mm] (4,-1.1 + 0.5*4) -- ++ (0.65,0); \filldraw[fill=black!80!white,draw=black!80!white] (4,-1.1 + 0.5*4) circle [radius=0.06cm]; 
\draw[-,draw=black!80!white,line width=0.25mm] (4.65cm,0.4cm) -- (4.65cm,0.9cm);

\draw[-,draw=black!80!white,line width=0.25mm] (4.65cm,0.4cm) -- ++ (1.05cm,0cm); \filldraw[fill=black!80!white,draw=black!80!white] (4.65cm,0.4cm) circle [radius=0.06cm];

\draw[-,draw=black!80!white,line width=0.25mm] (5.7,-1.1 + 0.5*3) -- ++ (0.7,0); \filldraw[fill=black!00!white,draw=black!80!white] (5.7,-1.1 + 0.5*3) circle [radius=0.07cm];

\draw[-,draw=black!80!white,line width=0.25mm] (6.4,-1.1 + 0.5*3) -- ++ (0.8,0); 
\draw[-,draw=black!80!white,line width=0.25mm] (7.2,-0.1cm) -- (7.2,0.4cm);

\draw[-,draw=black!80!white,line width=0.25mm] (7.2,-0.1) -- ++ (1.4,0); \filldraw[fill=black!80!white,draw=black!80!white] (7.2,-0.1) circle [radius=0.06cm]; 

\draw[-,draw=black!80!white,line width=0.25mm] (8.6,-0.1) -- ++ (0.5,0); \filldraw[fill=black!00!white,,draw=black!80!white] (8.6,-0.1) circle [radius=0.07cm]; 

\foreach \i in {5.7,8.6} {
\filldraw[fill=black!00!white,draw=black!80!white] (\i,-1.6) circle [radius=0.07cm]; 
}
\foreach \i in {4,4.65,7.2} {
\filldraw[fill=black!20!white,draw=black!80!white] (\i,-1.6) circle [radius=0.07cm]; 
}
\node at (4,-2.06) {\small $\hat{t}_1$};
\node at (4.65,-2.06) {\small $\hat{t}_2$};
\node at (5.7,-2.06) {\small $\hat{t}_3$};
\node at (7.2,-2.06) {\small $\hat{t}_4$};
\node at (8.6,-2.06) {\small $\hat{t}_5$};
\node at (3.2cm,0.75cm) {$\hat{x}_0$};
\node at (4cm,1.25cm) {$\hat{x}_1$};
\node at (4.65cm,0.0cm) {$\hat{x}_2$};
\node at (5.7cm,0.85cm) {$\hat{x}_3$};
\node at (7.2cm,-0.45cm) {$\hat{x}_4$};
\node at (8.6cm,0.25cm) {$\hat{x}_5$};
\end{tikzpicture}}\hfill
\resizebox{0.32\textwidth}{!}{%
\begin{tikzpicture}
\draw[->,line width=0.10mm] (3.2cm,-1.6cm) -- ++ (6.1,0cm);
\draw (3.2,-1.65cm) -- ++(0cm,0.1cm);
\foreach \i in {0,...,5} {
 \node at (2.6,-1.1 + 0.5*\i) {\small $\i$}; \draw[dashed,draw=black!10!white,line width=0.10mm] (3.2,-1.1 + 0.5*\i) -- ++ (6.1,0);
}

\draw[-,draw=black!80!white,line width=0.25mm] (3.2,0.4) -- ++ (0.8,0); \filldraw[fill=black!80!white,draw=black!80!white] (3.2,0.4) circle [radius=0.06cm]; 
\draw[-,draw=black!80!white,line width=0.25mm] (4,0.4) -- (4,0.9);

\draw[-,draw=black!80!white,line width=0.25mm] (4,-1.1 + 0.5*4) -- ++ (0.65,0); \filldraw[fill=black!80!white,draw=black!80!white] (4,-1.1 + 0.5*4) circle [radius=0.06cm]; 
\draw[-,draw=black!80!white,line width=0.25mm] (4.65cm,0.4cm) -- (4.65cm,0.9cm);

\draw[-,draw=black!80!white,line width=0.25mm] (4.65cm,0.4cm) -- ++ (1.05cm,0cm); \filldraw[fill=black!80!white,draw=black!80!white] (4.65cm,0.4cm) circle [radius=0.06cm]; \filldraw[fill=blue!20!white,draw=black!80!white] (5.1cm,0.4cm) circle [radius=0.09cm]; 

\draw[-,draw=black!80!white,line width=0.25mm] (5.7,-1.1 + 0.5*3) -- ++ (0.7,0); \filldraw[fill=black!00!white,draw=black!80!white] (5.7,-1.1 + 0.5*3) circle [radius=0.07cm]; \filldraw[fill=blue!20!white,draw=black!80!white] (6.2cm,0.4cm) circle [radius=0.09cm]; 

\draw[-,draw=black!80!white,line width=0.25mm] (6.4,-1.1 + 0.5*3) -- ++ (0.8,0);  \filldraw[fill=blue!20!white,draw=black!80!white] (6.7cm,0.4cm) circle [radius=0.09cm]; 
\draw[-,draw=black!80!white,line width=0.25mm] (7.2,-0.1cm) -- (7.2,0.4cm);

\draw[-,draw=black!80!white,line width=0.25mm] (7.2,-0.1) -- ++ (1.4,0); \filldraw[fill=black!80!white,draw=black!80!white] (7.2,-0.1) circle [radius=0.06cm]; 

\draw[-,draw=black!80!white,line width=0.25mm] (8.6,-0.1) -- ++ (0.5,0); \filldraw[fill=black!00!white,,draw=black!80!white] (8.6,-0.1) circle [radius=0.07cm]; 

\foreach \i in {5.1,5.7,6.2,6.7,8.6} {
\filldraw[fill=black!00!white,draw=black!80!white] (\i,-1.6) circle [radius=0.07cm]; 
}
\foreach \i in {4,4.65,7.2} {
\filldraw[fill=black!20!white,draw=black!80!white] (\i,-1.6) circle [radius=0.07cm]; 
}
\node at (4,-2.06) {\small $\hat{t}_1$};
\node at (4.65,-2.06) {\small $\hat{t}_2$};
\node at (5.1,-2.06) {\small $\hat{t}_3$};
\node at (5.7,-2.06) {\small $\hat{t}_4$};
\node at (6.2,-2.06) {\small $\hat{t}_5$};
\node at (6.7,-2.06) {\small $\hat{t}_6$};
\node at (7.2,-2.06) {\small $\hat{t}_7$};
\node at (8.6,-2.06) {\small $\hat{t}_8$};
\node at (3.2cm,0.75cm) {$\hat{x}_0$};
\node at (4cm,1.25cm) {$\hat{x}_1$};
\node at (4.65cm,0.0cm) {$\hat{x}_2$};
\node at (5.1cm,0.85cm) {$\hat{x}_3$};
\node at (5.7cm,0.85cm) {$\hat{x}_4$};
\node at (6.2cm,0.85cm) {$\hat{x}_5$};
\node at (6.7cm,0.85cm) {$\hat{x}_6$};
\node at (7.2cm,-0.45cm) {$\hat{x}_7$};
\node at (8.6cm,0.25cm) {$\hat{x}_8$};
\end{tikzpicture}}\hfill
\resizebox{0.32\textwidth}{!}{%
\begin{tikzpicture}
\draw[->,line width=0.10mm] (3.2cm,-1.6cm) -- ++ (6.1,0cm);
\draw (3.2,-1.65cm) -- ++(0cm,0.1cm);
\foreach \i in {0,...,5} {
 \node at (2.6,-1.1 + 0.5*\i) {\small $\i$}; \draw[dashed,draw=black!10!white,line width=0.10mm] (3.2,-1.1 + 0.5*\i) -- ++ (6.1,0);
}

\draw[-,draw=black!80!white,line width=0.25mm] (3.2,0.4) -- ++ (0.8,0); \filldraw[fill=black!80!white,draw=black!80!white] (3.2,0.4) circle [radius=0.06cm]; 
\draw[-,draw=black!80!white,line width=0.25mm] (4,0.4) -- (4,0.9);

\draw[-,draw=black!80!white,line width=0.25mm] (4,-1.1 + 0.5*4) -- ++ (0.65,0); \filldraw[fill=black!80!white,draw=black!80!white] (4,-1.1 + 0.5*4) circle [radius=0.06cm]; 
\draw[-,draw=black!80!white,line width=0.25mm] (4.65cm,0.4cm) -- (4.65cm,0.9cm);

\draw[-,draw=black!80!white,line width=0.25mm] (4.65cm,0.4cm) -- ++ (1.05cm,0cm);   \filldraw[fill=blue!20!white,draw=black!80!white] (4.65,-1.1 + 0.5*3) circle [radius=0.14cm]; \filldraw[fill=black!80!white,draw=black!80!white] (4.65cm,0.4cm) circle [radius=0.06cm];

\draw[-,draw=black!80!white,line width=0.25mm] (5.7,-1.1 + 0.5*3) -- ++ (0.7,0); \filldraw[fill=blue!20!white,draw=black!80!white] (5.7,-1.1 + 0.5*3) circle [radius=0.21cm]; \filldraw[fill=blue!20!white,draw=black!80!white] (5.7,-1.1 + 0.5*3) circle [radius=0.14cm]; \filldraw[fill=black!00!white,draw=black!80!white] (5.7,-1.1 + 0.5*3) circle [radius=0.07cm];

\draw[-,draw=black!80!white,line width=0.25mm] (6.4,-1.1 + 0.5*3) -- ++ (0.8,0); 
\draw[-,draw=black!80!white,line width=0.25mm] (7.2,-0.1cm) -- (7.2,0.4cm);

\draw[-,draw=black!80!white,line width=0.25mm] (7.2,-0.1) -- ++ (1.4,0); \filldraw[fill=black!80!white,draw=black!80!white] (7.2,-0.1) circle [radius=0.06cm]; 

\draw[-,draw=black!80!white,line width=0.25mm] (8.6,-0.1) -- ++ (0.5,0); \filldraw[fill=black!00!white,,draw=black!80!white] (8.6,-0.1) circle [radius=0.07cm]; 

\foreach \i in {5.7,8.6} {
\filldraw[fill=black!00!white,draw=black!80!white] (\i,-1.6) circle [radius=0.07cm]; 
}
\foreach \i in {4,4.65,7.2} {
\filldraw[fill=black!20!white,draw=black!80!white] (\i,-1.6) circle [radius=0.07cm]; 
}

\node at (4,-2.06) {\small $\hat{t}_1$};
\node at (4.65,-2.06) {\small $\hat{t}_2$};
\node at (5.7,-2.06) {\small $\hat{t}_3$};
\node at (7.2,-2.06) {\small $\hat{t}_4$};
\node at (8.6,-2.06) {\small $\hat{t}_5$};
\node at (3.2cm,0.75cm) {$\hat{x}_0$};
\node at (4cm,1.25cm) {$\hat{x}_1$};
\node at (4.65cm,0.0cm) {$\hat{x}_2$};
\node at (5.7cm,0.85cm) {$\hat{x}_3$};
\node at (7.2cm,-0.45cm) {$\hat{x}_4$};
\node at (8.6cm,0.25cm) {$\hat{x}_5$};
\end{tikzpicture}}
\vskip 0in
\caption{Schematic of augmentation procedure to $(\hat{\boldsymbol{t}},\hat{\boldsymbol{x}},\boldsymbol{k})$. Left, a trajectory $(\hat{\boldsymbol{t}},\hat{\boldsymbol{x}})$ with virtual jumps. Centre, \textit{compensating} virtual epochs are superimposed. Right, superimposed epochs assigned as weights; times ignored.} 
\label{samplingExample2}
\end{center}
\vskip -0.1in
\end{figure*}
On the left, we find an augmented trajectory $(\hat{\boldsymbol{t}},\hat{\boldsymbol{x}})$; this includes the \textit{real} MJP $(\boldsymbol{t},\boldsymbol{x})$ along with two virtual jumps sampled from a `\textit{controlled}' Poisson process with rate $\psi(\cdot)>0$. In the centre, further virtual epochs are added from a `\textit{compensating}' process. This joint procedure corresponds to steps 1-2 within Algorithm \ref{naiveAlgo}, and is equivalent to \textit{splitting} augmentation steps for stationary processes outlined in \cite{rao13a}, where a larger sequence $\hat{\boldsymbol{t}}$ is directly sampled from \eqref{virtualJumps}. In the right diagram, the compensating virtual epochs are re-assigned as weights $k_i$, $i\geq 0$ over their corresponding nodes; within Algorithm \ref{naiveAlgoUnif}, the times may be ignored for the purpose of re-sampling a new trajectory $\hat{\boldsymbol{x}}$. 

These diagrams help to depict major shortcomings behind traditional uniformization schemes \citep[see][]{hobolth2009,rao13a} for inference with stationary systems. Note that any (augmented) sequence $\hat{\boldsymbol{t}}$ is effectively a random discretization of a time-interval $[0,T]$, and serves as a basis for \textit{forward-backward} procedures. Yet, population models are always governed by large/infinite generator matrices $Q$, and are tied to large dominating rates $\Omega > \max_{x\in\mathcal{S}} |Q_x|$. This leads to sizeable candidate sets $\hat{\boldsymbol{t}}$ with associated overheads during forward-backward procedures. However, underling trajectories in $\mathcal{X}$ are unlikely to consistently transition states in $\mathcal{S}$ whose departure rates are `\textit{close}' to $\Omega$. Thus, the majority of candidate times in $\hat{\boldsymbol{t}}$ will require thinning anyway. As observed in Figure \ref{samplingExample2}, this paper builds over data augmentation techniques that restrict the cardinality of $\hat{\boldsymbol{t}}$, and correspondingly penalise self-transitions in order to preserve asymptotic exactness.

\subsection{Limiting properties and arbitrarily large bounds}

We begin with a preliminary result regarding convergence of sequences of random variables.

\begin{lemma} \label{lemmaL2}
Let $a\in\mathbb{R}$ and $b,c\in\mathbb{R}_{>0}$ be some fixed constant values, and define random variables $u_\kappa$ by non-linear transformations  $$u_\kappa = \Big( 1+\frac{a}{\kappa} \Big)^{v_\kappa}, \quad v_\kappa\sim\mathcal{P}\big([\kappa + b]\cdot c\big)$$ for all $\kappa\in\mathbb{R}_{>0}$, s.t. every $v_\kappa$ denotes a Poisson random variable with mean rate $(\kappa + b)\cdot c$. Then, $u_\kappa\xrightarrow{L^2} e^{a\cdot c}$  as $\kappa\rightarrow\infty$, and any sequence of random variables $u_i$, $i\in\mathcal{I}$ defined over an increasing and unbounded index set $I$ converges in mean square to the same constant value.
\end{lemma}
\begin{proof}
We show that $\mathbb{E}[u_\kappa^2]$ exists for all $\kappa\in\mathbb{R}_{>0}$, and $\lim_{\kappa\rightarrow\infty}\mathbb{E}\Big[\big(u_\kappa-e^{a c}\big)^2\Big]=0.$
First, note that
\begin{align*}
\mathbb{E}[u_\kappa^2]&=\sum_{x=0}^\infty \frac{[(\kappa+b)\cdot c]^x e^{-(\kappa+b) c}}{x!}\Big(1+\frac{a}{\kappa}\Big)^{2x}  = e^{2ac + (a^2c+2abc)/\kappa + a^2bc/\kappa^2},
\end{align*}
which is well defined for all $\kappa\in\mathbb{R}_{>0}$. Similarly $\mathbb{E}[u_\kappa] = e^{ac+ abc/\kappa}$, and it follows that
\begin{align*}
\mathbb{E}\Big[\big(u_\kappa-e^{a\cdot c}\big)^2\Big] = \mathbb{E}\Big[u_\kappa^2\Big] - 2\cdot e^{ac}\cdot\mathbb{E}\Big[u_\kappa\Big] + e^{2ac} \xrightarrow{\kappa\rightarrow\infty} 0. 
\end{align*}
\end{proof}

Next, for each epoch $i=0,\dots,m$ within an (augmented) sequence $\hat{\boldsymbol{t}}$, define a further partition of the interval $[\hat{t}_i,\hat{t}_{i+1}]$, into $\nu$ equally spaced subintervals of step size $\Delta t=\frac{\hat{t}_{i+1} - \hat{t}_{i}}{\nu}$, s.t. 
\begin{align}
\int_{\hat{t}_i}^{\hat{t}_{i+1}} [\Omega + Q_{\hat{x}_i}(s) - \psi(s,\hat{x}_i)] \mathrm{d}s \approx \sum_{j=0}^{\nu-1} \Delta t \cdot [\Omega + Q_{\hat{x}_i}(\hat{t}_i + j\cdot \Delta t) - \psi(\hat{t}_i + j\cdot \Delta t,\hat{x}_i)] \label{RiemannApprox}
\end{align}
offers a Riemann approximation (exact as $\Delta t \rightarrow 0$) to the intensity of \textit{compensating} jumps in Proposition \ref{propDataAug} and variables $k_i$ in Algorithm \ref{naiveAlgo} (Step 2). The approximating rate is piecewise constant; s.t. \textit{compensating} jumps under \eqref{RiemannApprox} are uniformly distributed in each tagged subinterval $j=0,\dots,\nu-1$ of $[\hat{t}_i,\hat{t}_{i+1}]$. Thus, for all $i=0,\dots,m$ Poisson counts $k_i^j$ respond to rates $\Delta t \cdot [\Omega + Q_{\hat{x}_i}(\hat{t}_i + j\cdot \Delta t) - \psi(\hat{t}_i + j\cdot \Delta t,\hat{x}_i)]$; and $w_i(x)$ in \eqref{weightsDef} is approximated by
\begin{align*}
w_i(x)\approx\prod_{j=0}^{\nu-1} \bigg(1+\frac{Q_{\hat{x}_i}(\hat{t}_i + j\cdot \Delta t) - \psi(\hat{t}_i + j\cdot \Delta t,\hat{x}_i)}{\Omega}\bigg)^{k_i^j}. 
\end{align*} 
By Lemma \ref{lemmaL2}, as the dominating rate $\Omega\rightarrow \infty$, and thus the inferential framework accommodates arbitrarily large rates within $\boldsymbol{Q}(\boldsymbol{\lambda})$, it further holds
\begin{align*}
w_i(x)\approx \exp(\sum_{j=0}^{\nu-1} \Delta t \cdot [Q_{\hat{x}_i}(\hat{t}_i + j\cdot \Delta t) - \psi(\hat{t}_i + j\cdot \Delta t,\hat{x}_i)]).
\end{align*} 
By finally taking the limit $\Delta t \xrightarrow{\nu\rightarrow \infty} 0$ to retrieve the original integral representation in \eqref{RiemannApprox}, we have $w_i(x) =  \exp\big(\int_{\hat{t}_i}^{\hat{t}_{i+1}}  [Q_{\hat{x}_i}(s) - \psi(s,\hat{x}_i)] \mathrm{d}s\big)$, which leads to a simplified sampler design for non-stationary systems as shown in Algorithm \ref{algoMainMJP} (similarly amendable to observations). 
\begin{algorithm}[t]
  \normalsize
\caption{Reduced construction of non-stationary correlated MJP trajectories on $\mathcal{X}$.}
\label{algoMainMJP}
\vspace{2pt}
\begin{minipage}[t]{0.08\textwidth}
\textit{Input:}
\end{minipage}
\begin{minipage}[t]{0.9\textwidth}
Sequence of intensity matrices  $\boldsymbol{Q}=\boldsymbol{Q}(\boldsymbol{\lambda})$ parametrized by $\boldsymbol{\lambda}$. \\[1pt]
An MJP trajectory $(\boldsymbol{t},\boldsymbol{x})\in\mathcal{X}$ with $\boldsymbol{t}=\{t_0,\dots,t_{n}\}$ and $\boldsymbol{x}=\{x_0,\dots,x_{n}\}$. \\[1pt]
Arbitrary intensity operator $\psi:[0,T]\times\mathcal{S}\rightarrow\mathbb{R}_+$ for candidate times.
\end{minipage}
\vspace{4pt}

\begin{minipage}[t]{0.08\textwidth}
\textit{Output:}
\end{minipage}
\begin{minipage}[t]{0.9\textwidth}
A new MJP trajectory $(\boldsymbol{t},\boldsymbol{x})_{new}\in\mathcal{X}$ sampled from the density $f_{X}(\boldsymbol{t},\boldsymbol{x}|\boldsymbol{Q})$ in \eqref{pathProbs}.
\end{minipage}
\vspace{-6pt}

\hrulefill
\begin{algorithmic}[1] 
\State Create an (ordered) set of candidate times $\hat{\boldsymbol{t}}=\{\hat{t}_0,\dots,\hat{t}_{m}\}$, $m\geq n$, attaching to $\boldsymbol{t}$ auxiliary events from a Poisson process; rate $\psi(t,x_i)>0$, within intervals $(t_i,t_{i+1})$, with $t_{n+1}=T$. \vspace{2pt}
\State Draw a new sequence $\hat{\boldsymbol{x}}=\{\hat{x}_0,\dots,\hat{x}_{m}\}$ with a forward-backward procedure; given initial distribution $\pi(x)$, transition \textit{weight} matrices
\begin{align*}
\tilde{P}(\hat{t}_i) = \textnormal{diag}( \{\psi(\hat{t}_i,x)-Q_{x}(\hat{t}_i) \, : \, x\in\mathcal{S}\}) + Q(\hat{t}_i) , 
\end{align*}
and importance weights
\begin{align*}
w_i(x) =  \exp(\int_{\hat{t}_i}^{\hat{t}_{i+1}}  [Q_{\hat{x}_i}(s) - \psi(s,\hat{x}_i)] \mathrm{d}s),
\end{align*} 
imposed over epochs $i\in\{0,\dots,m\}$.  \vspace{2pt}
\State Remove self-transitions on $(\hat{\boldsymbol{t}},\hat{\boldsymbol{x}})$ to produce $(\boldsymbol{t},\boldsymbol{x})_{new}$.
\end{algorithmic}
\end{algorithm}
There, note that the dominating rate $\Omega$ and \textit{compensating} jumps are no longer relevant. The result retrieves an analogue construction to algorithmic propositions for semi-Markov processes in \cite{rao2012mcmc}; however, it requires iterative calculations of exponential functionals, notoriously resource-demanding in computational implementations.

\section{Scalable sampling of deviations from mean-average dynamics} \label{ODEsection}

We continue with novel integrations of auxiliary variables to sample MJP paths in $f_{X}(\boldsymbol{t},\boldsymbol{x}|\boldsymbol{\lambda},\boldsymbol{O})$ in \eqref{targetdens} as \textit{controlled} deviations from approximate mean-average dynamics. As reference, we use a time functional $\xi(t)_{0\leq t\leq T}$ supported on an aribitrary set, so that a \textit{distance} to population levels in $\mathcal{S}$ is quantifiable. We require that $\xi(t)$ is \textit{close} to a region of high density in the posterior distribution of $X_t|\boldsymbol{O}$, for all $t\in[0,T]$. Under reasonably mild conditions, limiting theorems in \cite{kurtz_1970,kurtz_1971} guarantee that every stochastic jump process accepts a real-valued deterministic approximation, as a solution to a system of \textit{ordinary differential equations} (ODEs). This can be further calibrated to observed data in a computationally inexpensive manner, and we observe an example on the left hand side diagram within Figure \ref{fig:AuxVar}. 
\begin{figure}[h!]
  \centering
   \includegraphics[width=\linewidth]{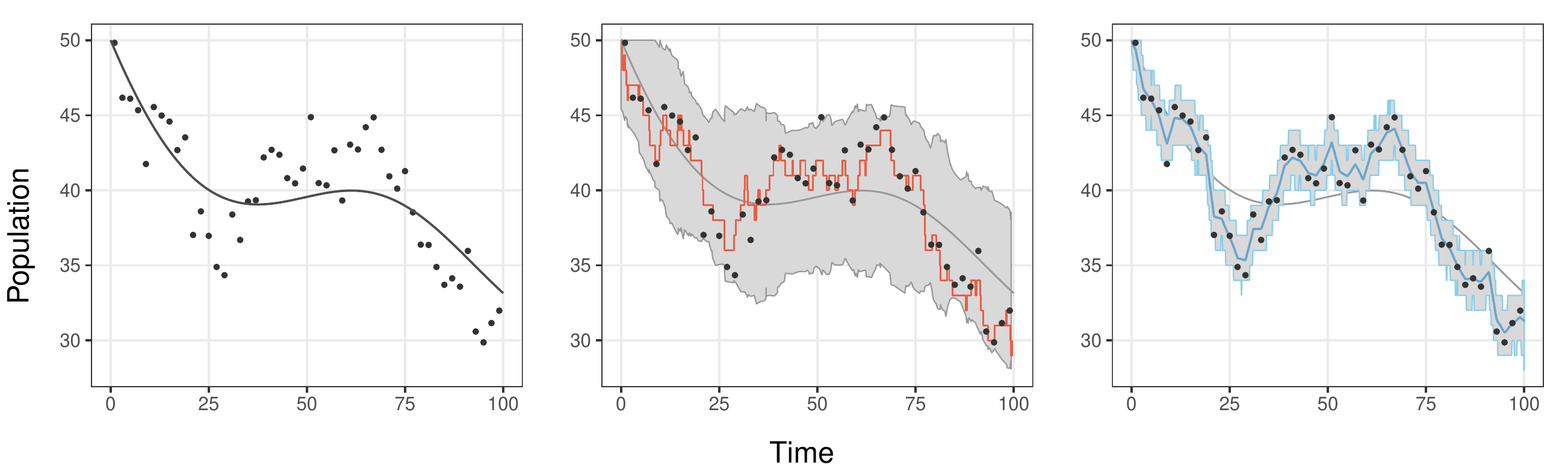}
  \caption{Left, noisy observations of \textit{birth-death} population levels, along with a calibrated solution to an ODE. Centre, sample trajectory from the posterior MJP (in red), restricted to a (random) subspace of $\mathcal{X}$ defined by a range (in grey) centred at $\xi(t)_{0\leq t\leq T}$. Right, posterior mean-average path and $95\%$ confidence interval for $X$.}
  \label{fig:AuxVar}
\end{figure}

Given $\xi(t)_{0\leq t\leq T}$, we complement each iteration in Algorithms \ref{naiveAlgo}-\ref{algoMainMJP} with a further auxiliary sequence $\boldsymbol{u}$ of the form \eqref{intAux}. The goal is to form an informative set that restricts the explorable space of $\hat{\boldsymbol{x}}$ in sampling steps for \eqref{StSpRep}. Importantly, this must be completed within Gibbs procedures and thus not compromise the mixing properties of the MCMC sampler \cite[cf.][]{georgoulas2017unbiased}. For simplicity in the presentation, we restrict the following formulations to integer-valued univariate population systems, where $\xi(t)_{0\leq t\leq T}$ is a real-valued function; however, the various definitions are readily amendable to multivariate models with various support sets. 

\noindent\textbf{Auxiliary truncated normal random variables.} For a current (augmented) trajectory $(\hat{\boldsymbol{t}},\hat{\boldsymbol{x}})$, define
\begin{align}
g(u_i|u_{i-1},\hat{x}_i,\xi(\hat{t}_i)) = \phi\bigg(\frac{u_i-\mu_i}{\sigma}\bigg)\bigg/\sigma\bigg[1-\Phi\bigg(\frac{|\hat{x}_i-\xi(\hat{t}_i)|- \mu_i}{\sigma}\bigg)\bigg], \label{normalAux}
\end{align}
whenever $\hat{x}_i\in\big(\xi(\hat{t}_i)-u_i,\xi(\hat{t}_i)+u_i\big)$, $i=1,\dots,m$; where 
$$\mu_i = \max(\mu, u_{i-1} - \kappa), \quad \text{with} \quad \mu_1=\mu \in\mathbb{R}_+, \, \kappa\in(0,1),$$
and $\phi(\cdot)$, $\Phi(\cdot)$ denote the standard normal density/cumulative distribution functions, respectively. Each $u_i$ is thus normally distributed (mean $\mu_i$, standard deviation $\sigma$) and truncated to a space $(|\hat{x}_i-\xi(\hat{t}_i)|,\infty)\in\mathbb{R}$. The minimum deviation a newly sampled sequence $\hat{\boldsymbol{x}}$ will, on average, be allowed to distance itself from $\xi(\cdot)$ is defined by $\mu$; and $\kappa$ accommodates mean reverting dynamics as depicted in Figure \ref{fig:AuxVardiag}. We find an example in the centre diagram within Figure \ref{fig:AuxVar}; there, the greyed area denotes the region between lower/upper boundaries $\xi(\hat{t}_i)-u_i$ and $\xi(\hat{t}_i)+u_i$, across epochs $i=1,\dots,m$. 
\begin{figure}[h!] \vspace*{-5pt}
\centering
\begin{tikzpicture}
    \hspace*{-10pt}\draw (0, 0) node[inner sep=0] {\includegraphics[width=\linewidth]{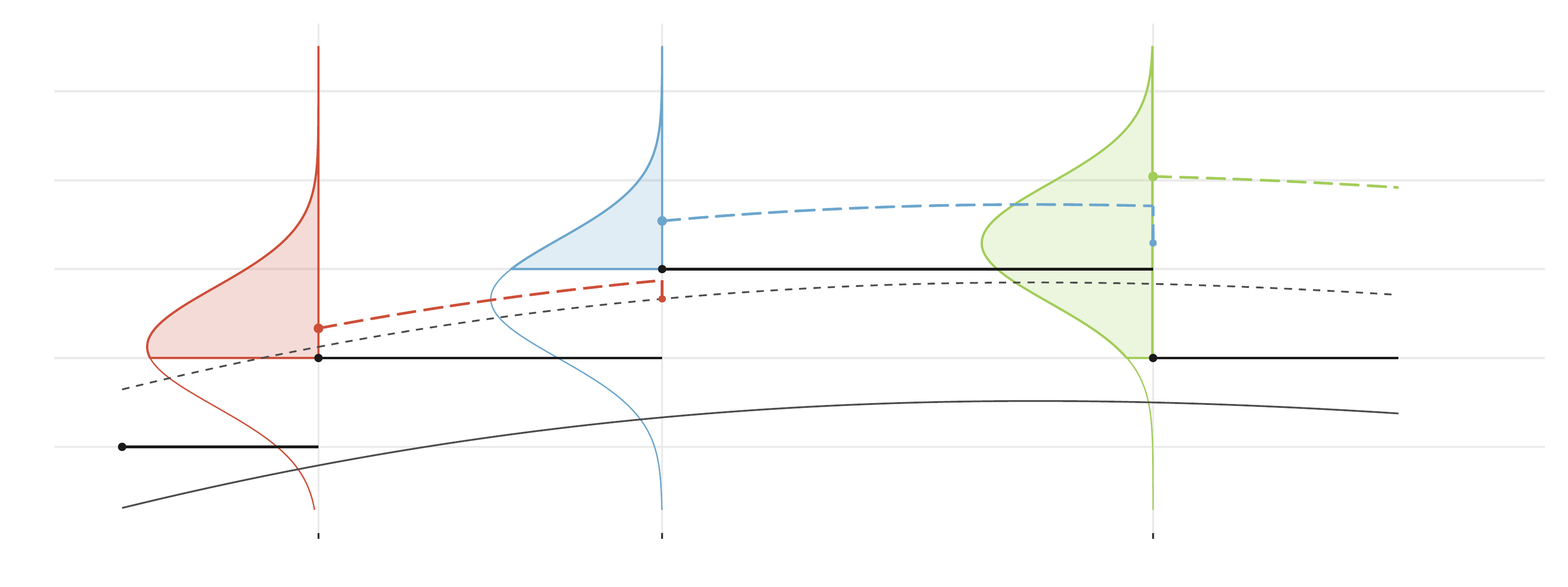}};
    \draw (-1.15, -2.7) node {\footnotesize $\hat{t}_i$}; 
    \draw (-4.5, -2.7) node {\footnotesize $\hat{t}_{i-1}$}; 
    \draw (3.7, -2.7) node {\footnotesize $\hat{t}_{i+1}$}; 
	\draw (6.8, -1.2) node {\footnotesize $\xi(t)_{0\leq t\leq T}$}; 
	\draw (7.1, 0) node {\footnotesize $\mu + \xi(t)_{0\leq t\leq T}$}; 
	\draw (-1.98, 0.76) node {\footnotesize $\xi(\hat{t}_i) + u_i$}; 
	\draw (-0.55, -0.25) node {\footnotesize $\mu + \xi(\hat{t}_i)$}; 
	\draw (-5.7, -0.3) node {\footnotesize $\xi(\hat{t}_{i-1}) + u_{i-1}$};
	\draw (2.5, 1.2) node {\footnotesize $\xi(\hat{t}_{i+1}) + u_{i+1}$}; 
	\draw (4.9, 0.5) node {\footnotesize $\xi(\hat{t}_{i+1}) + u_{i} - \kappa$}; 
\end{tikzpicture}
  \caption{Sample sketch with auxiliary normal random variables (coloured dots), superimposed to an (augmented) sequence $\hat{\boldsymbol{x}}$ (in black) above the mean-average dynamics $\xi(t)_{0\leq t\leq T}$. Shaded coloured areas represent the truncated densities associated with random deviations $u_i$ from $\xi(\hat{t}_i)$, $i\geq 0$. }
  \label{fig:AuxVardiag}
\end{figure}
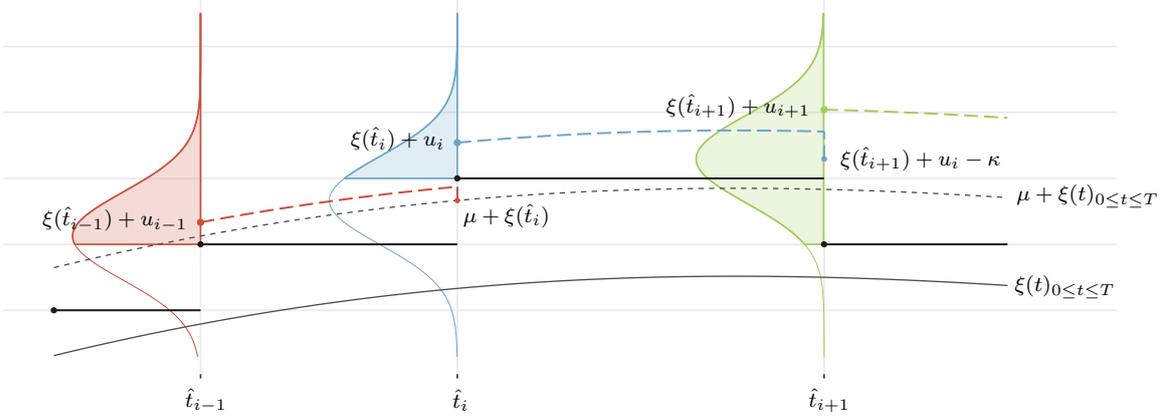

In order to sample a new \textit{compatible} sequence $\hat{\boldsymbol{x}}|\hat{\boldsymbol{t}},\boldsymbol{u}$  within Gibbs steps in Algorithm \ref{naiveAlgo}, forward filtering procedures with matrices \eqref{transPdef} correspond to equations
\begin{equation}
\mathbb{P}(\hat{x}_i=x|u_1,\dots,u_{i};\hat{\boldsymbol{t}}) \propto \mathbb{I}\big(x\in\bar{\mathcal{S}}_i\big) \cdot w_i(x) \cdot \sum_{x'\in\bar{\mathcal{S}}_{i-1}} \tilde{P}_{x',x}(\hat{t}_i) \cdot \mathbb{P}(\hat{x}_{i-1}=x'|u_1,\dots,u_{i-1};\hat{\boldsymbol{t}}) \label{eqFF}
\end{equation}
for $i=1,\dots,m$, with $\bar{\mathcal{S}}_i = \{ x\in\mathcal{S} \, : \, \xi(\hat{t}_i)-u_i \leq x \leq \xi(\hat{t}_i)+u_i   \}$ and importance weights
$$ w_i(x)= \prod_{s\in\boldsymbol{s}_i} \bigg(1+\frac{Q_x(s)-\psi(s,x)}{\Omega}\bigg) \bigg/ \Phi\bigg(\frac{\mu_i - |x-\xi(\hat{t}_i)|}{\sigma}\bigg).$$
Backward sampling steps remain unaltered, s.t. $\hat{x}_m$ is sampled from within $\bar{\mathcal{S}}_m$ in proportion to $\mathbb{P}(\hat{x}_m|\boldsymbol{u};\hat{\boldsymbol{t}})$; then, for $i = m-1,\dots,0$ we may sample subsequent states within $\bar{\mathcal{S}}_i$ from
\begin{align}
\mathbb{P}(\hat{x}_{i}|\hat{x}_{i+1},\boldsymbol{u};\hat{\boldsymbol{t}}) \propto \tilde{P}_{\hat{x}_{i},\hat{x}_{i+1}}(\hat{t}_{i+1}) \cdot \mathbb{P}(\hat{x}_i|u_1,\dots,u_{i};\hat{\boldsymbol{t}}). \label{backSteps}
\end{align}
The computational burden of the algorithm is thus restricted to calculations of quadratic complexity over an statistically \textit{controllable} space. To further condition a trajectory on observations $\boldsymbol{O}$ (pictured in red within centre diagram in Figure \ref{fig:AuxVar}), matrices $\tilde{P}$ and weights $w(\cdot)$ are altered according to definitions in Subsection \ref{twoSetpSubsec}. Analogue predictive/update steps to incorporate these truncation techniques within Algorithms \ref{naiveAlgoUnif}-\ref{algoMainMJP} follow naturally. Finally, in order to efficiently obtain (unbiased) estimates of the posterior trajectory of a population model (rightmost diagram in Figure \ref{fig:AuxVar}), we alternate between: (i) define subsets of $\mathcal{X}$ centred around $\xi(\cdot)$ and (ii) produce new trajectories $\hat{\boldsymbol{x}}$ within.

\noindent\textbf{Auxiliary Gamma random variables.} In this variant, suited in combination with population models subject to jumps of unit length, we let
\begin{eqnarray}
u_i = |x_i-\xi(\hat{t}_i)| + v_i \quad \text{with} \quad v_j\sim\Gamma(\alpha,\beta_i), \label{gammaVars}
\end{eqnarray} 
over a subset of epochs $i\in\mathcal{I}\subseteq\{1,\dots,m\}$; i.e. auxiliary variables are undefined for $i=1,\dots,m$ s.t. $i\not\in\mathcal{I}$, and $v_i$, $i\in\mathcal{I}$ are gamma distributed random variables. Here, $\alpha\in\mathbb{N}$ will secure a fast evaluation of the associated densities; and rate parameters are subordinated to a random autoregressive process $\boldsymbol{\mu}$ (stationary mean $\mu$, lag-$1$ deviation $\sigma$), s.t. $\beta_i=\alpha\cdot e^{-\mu_i}$ for all $i\in\mathcal{I}$, with
$$\mu_i|\mu_{i-l} \sim \mathcal{N}\Big(\mu\ + (\mu_{i-l}-\mu)(1-\kappa)^l \, , \, \sigma^2 \cdot \big(1-(1-\kappa)^{2l}\big)\big/\big(1-(1-\kappa)^2\big) \Big)\, ,$$
where $l\in\mathbb{N}$ denotes the lag between subsequent time points in $\mathcal{I}$, and $\kappa\in(0,1)$. Thus, $\mu_i \sim \mathcal{N}(\kappa\mu + (1-\kappa)\mu_{i-1},\sigma^2)$ whenever $l=1$ and $\mathcal{I}=\{1,\dots,m\}$ is associated to all times in $\hat{\boldsymbol{t}}$. This construct ensures $\mathbb{E}[u_i|x_i,\xi(\cdot)] = |x_i-\xi(\hat{t}_i)| + e^{\mu_i}$ and $\mathbb{V}[u_i] = e^{2\mu_i}/\alpha$; well calibrated, it allows for $\hat{\boldsymbol{x}}$ to significantly deviate from $\xi(t)_{0\leq t\leq T}$ over restricted time-intervals. A diagram depicting such structure of auxiliary variables is shown in Figure \ref{fig:AuxVardiag2}; there, $u_i$, $i\in\mathcal{I}$ are represented by coloured dots, placed over equally spaced epochs with lag $l=2$; means assigned to Gamma variables (grey dots) are random and transition according to log-normal distributions. 
\begin{figure}[!h] \vspace*{-5pt}
\centering
\begin{tikzpicture}
    \hspace*{-10pt}\draw (0, 0) node[inner sep=0] {\includegraphics[width=\linewidth]{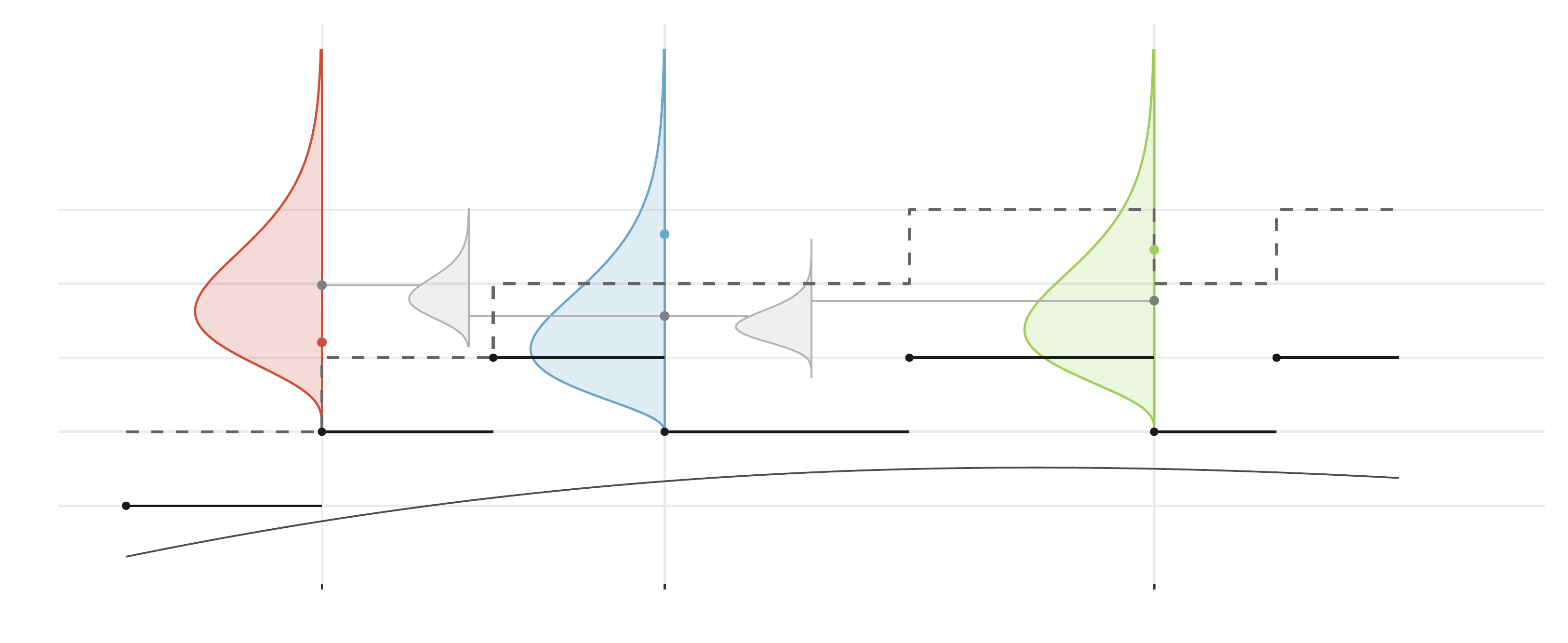}};
    \draw (-1.15, -2.9) node {\footnotesize $\hat{t}_i$}; 
    \draw (-4.5, -2.9) node {\footnotesize $\hat{t}_{i-2}$}; 
    \draw (3.7, -2.9) node {\footnotesize $\hat{t}_{i+2}$}; 
	\draw (6.8, -1.58) node {\footnotesize $\xi(t)_{0\leq t\leq T}$}; 
	\draw (6.9, 1.1) node {\footnotesize $\{\max\bar{\mathcal{S}}_i\}_{i\geq0}$}; 
	\draw (-1.98, 0.9) node {\footnotesize $\xi(\hat{t}_i) + u_i$}; 
	\draw (-5.7, -0.18) node {\footnotesize $\xi(\hat{t}_{i-2}) + u_{i-2}$};
	\draw (2.5, 0.75) node {\footnotesize $\xi(\hat{t}_{i+2}) + u_{i+2}$}; 
	\draw [decorate,decoration={brace,amplitude=3pt,mirror}] (-1.11,-1.05) -- (-1.11,0.02) node [black,midway,xshift=0.4cm] {\footnotesize $e^{\mu_i}$};	
	\draw [decorate,decoration={brace,amplitude=3pt,mirror}] (3.72,-1.05) -- (3.72,0.19) node [black,midway,xshift=0.55cm] {\footnotesize $e^{\mu_{i+2}}$};	
\end{tikzpicture}
  \caption{Sketch with auxiliary gamma variables (coloured dots), superimposed to a trajectory $\hat{\boldsymbol{x}}$ (in black) above $\xi(t)_{0\leq t\leq T}$. Coloured areas represent densities associated with $v_i, \, i\in\mathcal{I}$ over lagged epochs ($l=2$). Dashed line is the maximum deviation from $\xi(\hat{t}_i)$ that a newly (augmented) MJP can reach at times $\hat{t}_i, \, i\geq 0$. }\label{fig:AuxVardiag2}
\end{figure}
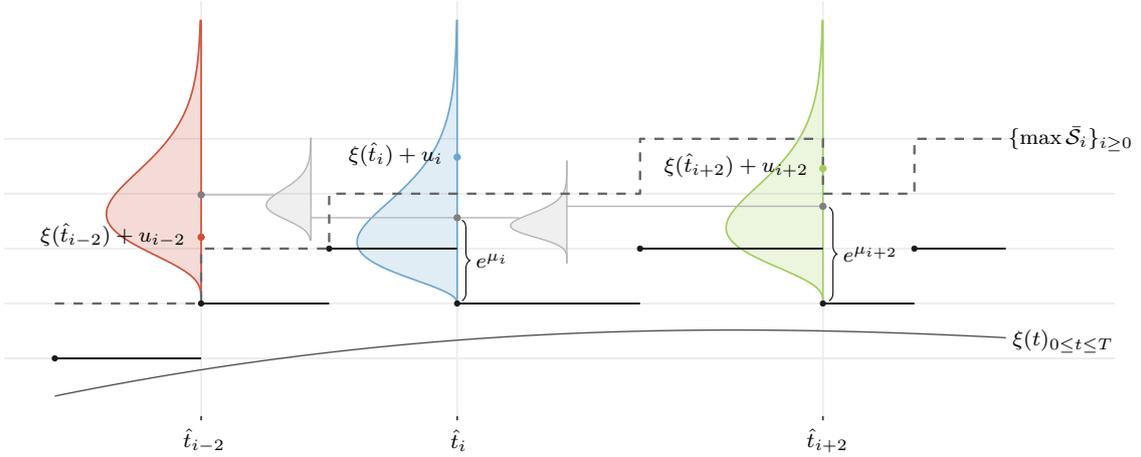

Next, assume process jumps are of unit length. In order to sample a \textit{compatible} sequence $\hat{\boldsymbol{x}}|\hat{\boldsymbol{t}},\boldsymbol{u}$ within Gibbs steps in Algorithm \ref{naiveAlgo}, the analogue to forward filtering procedures in \eqref{eqFF} is given by
\begin{align*}
\mathbb{P}(\hat{x}_i=x|\{u_j : & \, j\leq i, j\in\mathcal{I}\};\boldsymbol{\mu},\hat{\boldsymbol{t}}) \propto \\
&\mathbb{I}\big(x\in\bar{\mathcal{S}}_i\big) \cdot w_i(x|\boldsymbol{\mu}) \cdot \sum_{x'\in\bar{\mathcal{S}}_{i-1}} \tilde{P}_{x',x}(\hat{t}_i) \cdot \mathbb{P}(\hat{x}_{i-1}=x'|\{u_j : j< i, j\in\mathcal{I}\};\boldsymbol{\mu},\hat{\boldsymbol{t}}) 
\end{align*}
for $i=1,\dots,m$, with restricted subsets defined s.t.
\begin{align*}
\bar{\mathcal{S}}_i = \begin{cases}
    \{ x\in\mathcal{S} \, : \, \xi(\hat{t}_i)-u_i \leq x \leq \xi(\hat{t}_i)+u_i   \} & \text{if } i\in\mathcal{I},\\
    \{ x\in\mathcal{S} \, : \, \min\bar{\mathcal{S}}_{i-1} - 1 \leq x \leq \max\bar{\mathcal{S}}_{i-1} + 1   \} & \text{otherwise.}
  \end{cases}
\end{align*}
This assumes that $X$ is supported over an unbounded set of integers (but may be suitably redefined otherwise). Importance weights are given by
$$ w_i(x|\boldsymbol{\mu})= (u_i-|x-\xi(\hat{t}_i)|)^{\alpha_i-1} e^{-\beta_i (u_i - |x-\xi(\hat{t}_i)|)} \prod_{s\in\boldsymbol{s}_i} \bigg(1+\frac{Q_x(s)-\psi(s,x)}{\Omega}\bigg),$$
whenever $i\in\mathcal{I}$ and $ w_i(x|\boldsymbol{\mu})= \prod_{s\in\boldsymbol{s}_i} (1+(Q_x(s)-\psi(s,x))/\Omega)$ otherwise. This suggests a forward implementation with dynamic vectors, since the explorable space of MJP trajectories expands across jump epochs $i\not\in\mathcal{I}$, while contracting again towards mean-average dynamics in the presence of auxiliary evidence. Backward sampling steps still correspond to \eqref{backSteps} above. Again, similar amendments may be realized over Algorithms \ref{naiveAlgoUnif}-\ref{algoMainMJP}; also, further conditioning this procedure on observations corresponds to including alterations on $\tilde{P}$, $w(\cdot)$ as listed in Subsection \ref{twoSetpSubsec}.

Below, we discuss results of algorithmic implementations of these methods, on two instances of popular jump processes, and we draw comparisons on efficiency with current benchmark methodologies for inferential tasks. Results and comparisons reported are produced by C++ implementations; code and data can be found on \href{https://github.com/IkerPerez/scalableSamplingMJPs}{github.com/IkerPerez/scalableSamplingMJPs}.

\subsection{Example 1: A pure birth-death process}

A \textit{birth-death} process is a population model with applications in queueing theory and performance engineering tasks. In its simplest form, $X$ refers to a \textit{population} supported within the set of non-negative integers $\mathcal{S}=\mathbb{N}_0$, and state transitions involve both \textit{births} and \textit{deaths}. Infinitesimal rates for jumps are denoted by $\{\lambda_x(t)\}_{x\in\mathcal{S}}$ and $\{\mu_x(t)\}_{x\in\mathcal{S}}$, respectively, for all $t\geq 0$, so that
\begin{align*}
Q_{x,x'}(t) &= 
  \begin{cases}
    \lambda_x(t) & \text{if} \quad x'=x+1,\\
    \mu_x(t) \cdot \mathbb{I}(x>0) & \text{if} \quad x'=x-1, \\
    - \lambda_x(t) - \mu_x(t) \cdot \mathbb{I}(x>0) &  \text{if} \quad x'=x,
  \end{cases}
\end{align*}
and $Q_{x,x'}(t)=0$ otherwise. Hence, the process increases its population by $1$ whenever a \textit{birth} occurs; alternatively, it decreases its population by $1$ during a \textit{death} event. 

\noindent\textbf{A finite capacity immigration-death process.} In this variant, $\mathcal{S}$ is bounded from above by some positive constant $N\in\mathbb{N}_0$; thus, it is a system equivalent to a \textit{closed} queueing network with an infinite processor \citep[cf.][]{perez2018approximate}, or a \textit{truncated} M/M/N/N queue \citep{Gross:2008:FQT:1972549}. Importantly, for all states $x\in\{0,\dots,N\}$, death rates scale along with population levels, s.t. $\mu_x(t) = x\cdot \mu(t)$ for some time dependent function $\mu(\cdot)$. Here, we assume that arrivals enter the system with a constant birth intensity $\lambda_x(t)=\lambda \cdot \mathbb{I}(x<N)$ , $\lambda\in\mathbb{R}_+$; and death rates respond to seasonal patterns, s.t. $\mu(t)=\mu \cdot r(t)$ for some positive functional $r(t)\in[1,2], \, t\geq 0$. We further assume that $x_0=N$, and note that the model is fully parametrized by $\lambda$ and $\mu$.

\noindent\textbf{Noisy state observations and inference.} Let $\boldsymbol{O}=\{O_r\}_{r\geq 1}$ be state observations subject to measurement error, s.t. $O_r\sim\mathcal{N}(X_{t_r},\sigma^2)$ reflect normal random variables at times $t_r\in[0,T]$, $r\geq 1$. This, along with a finite population set-up, allows for the implementation (for comparison purposes) of benchmark uniformization-based inferential techniques. We find sample observations within the left diagram in Figure \ref{fig:AuxVar} (black dots), for a latent process realisation with capacity $N=50$ and seasonality $r(t)=3/2 + \cos(2\pi\cdot t/T)/2$. The dark line in the figure corresponds to the deterministic approximation $\xi(\cdot)$ with death rate parameter
\begin{align*}
\mu = \arg\min_{\mu\in\mathbb{R}_+} \sum_{r\geq 1}(\xi(t_r)-O_{r})^2 \quad \text{subject to} \quad
\frac{\mathrm{d}\xi(t)}{\mathrm{d}t} = \mathbb{I}(\xi(t)<50) \cdot \lambda - r(t) \cdot \xi(t) \cdot \mu,
\end{align*}
and a (known) birth rate $\lambda$ fixed to an arbitrary value (ensuring model identifiability).  

From \eqref{intractLikel}, notice that the posterior density $f(\mu|\boldsymbol{O})$ requires integrating the observation likelihood, over all possible trajectories with associated density
\begin{align*}
f_{X}(\boldsymbol{t},\boldsymbol{x}|\boldsymbol{Q})  & =  \pi(x_0) e^{-\sum_{i=1}^n\int_{t_{i-1}}^{t_i} (x_nr(s)\mu+\mathbb{I}(x_n<N)\lambda)    \mathrm{d}s} \prod_{i=1}^n \lambda^{\mathbb{I}(x_i=x_{i-1}+1)} [\mu\cdot x_{i-1}\cdot r(t_i)]^{\mathbb{I}(x_i=x_{i-1}-1)}, 
\end{align*}
and is thus intractable. In this task, we carry posterior MCMC inference on the death rate by iterating between sampling trajectories and $\mu$ from its conditional density 
$$f(\mu|\boldsymbol{t},\boldsymbol{x}) \propto \mu^{\sum_{i=1}^n \mathbb{I}(x_i=x_{i-1}-1)} e^{-\mu \sum_{i=1}^n x_n\int_{t_{i-1}}^{t_i} r(s)   \mathrm{d}s}  \cdot \pi(\mu),$$
with some loosely uninformative prior $\pi(\mu)$. A sample trace output is shown in Figure \ref{fig:traceBD}, corresponding to the data displayed within Figure \ref{fig:AuxVar}. There, the $3$ different traces and densities correspond to (i) a traditional uniformization-based implementation \citep{rao13a}, (ii) Algorithm \ref{algoMainMJP} and (iii) a variant centred around mean-average dynamics and normal auxiliary variables in \eqref{normalAux}. All alternatives yield equivalent estimates for $\mu$, of a seemingly similar quality.
\begin{figure}[h!]
  \centering
   \includegraphics[width=\linewidth]{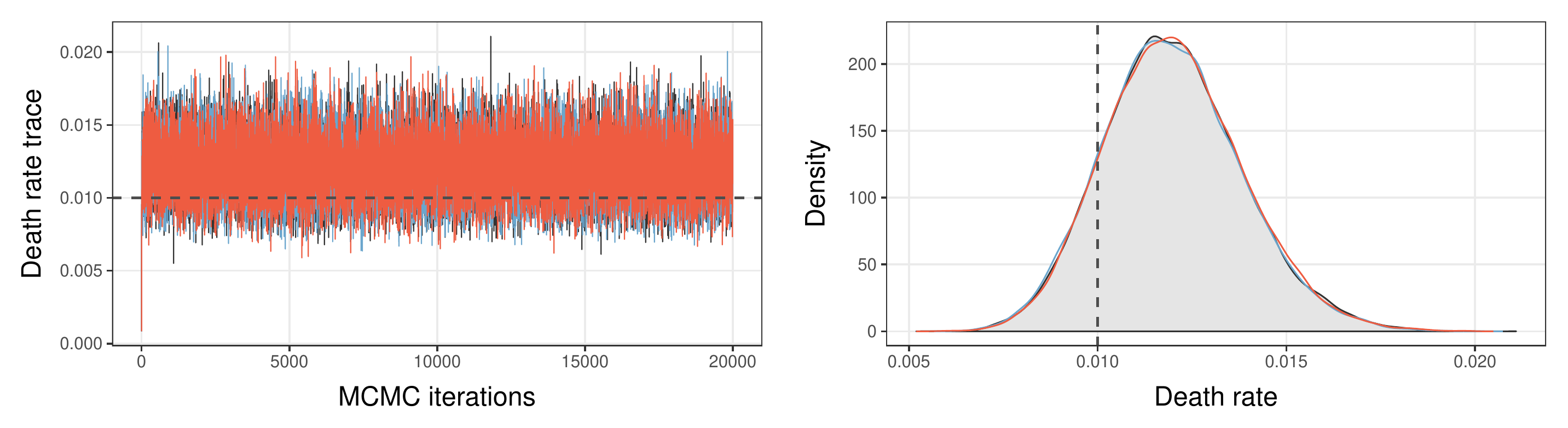}
  \caption{MCMC traces and kernel density representations for death rates, corresponding to uniformization (red), Algorithm \ref{algoMainMJP} (blue) and Algorithm \ref{algoMainMJP} combined with variables \eqref{normalAux} (black).}
  \label{fig:traceBD}
\end{figure}

We repeat the process for multiple simulated \textit{birth-death} trajectories, at increasing population sizes $N$. In all cases, $T=100$ and we produce $50$ noisy observations over equally spaced intervals. Comparisons on efficiency across various methods are offered in Figure \ref{fig:BDratios}, which further includes a summary of population sizes tested, birth rates used and deviation associated with the observations.
In the diagram, we find ratios in effective sample sizes (scaled for computation time) against the benchmark algorithm of \cite{rao13a} (black line) with dominating rate $\Omega = 1.5 \cdot\max_{x\in\mathcal{S}} \sup_{t\in[0,T]}|Q_x(t)|$. Whiskers represent $95\%$ confidence intervals. There, (i) the blue line corresponds to Algorithm \ref{algoMainMJP}, with $\psi(t,x) = |Q_x(t)|$, $(t,x)\in[0,T]\times\mathcal{S}$, s.t. \textit{auxiliary} jumps attached to $\boldsymbol{t}$ are generated in proportion to diagonal elements of $Q(t), \, t\in[0,T]$, and $\hat{\boldsymbol{t}}$ is of approximately double the size of $\boldsymbol{t}$ in each MCMC iteration, (ii) the green line is for further restricting MJP samples to deviations from $\xi(\cdot)$, using auxiliary variables \eqref{normalAux} with mean $\mu=N/10$, autoregressive coefficient $\kappa=1$ and Gaussian deviation $\sigma=0.65\cdot(1+\kappa)$; this offers a good heuristic, noting that birth-death jumps are of magnitude one s.t. truncated normal densities are always substantial, (iii) the red line finally assigns $\mu=\sqrt{N}$, $\kappa = 0.05$ and $\sigma=1.5\cdot(1+\kappa)$, s.t. explorable spaces for $\hat{\boldsymbol{x}}$ are very restricted around the mean-average solution, yet, randomness and autoregressive effects are strong and can accommodate sudden short-timed deviations from $\xi(t)_{t\in[0,T]}$.
\begin{figure}[h!]
\centering
\begin{minipage}{0.65\textwidth}
\centering
\includegraphics[width=\linewidth]{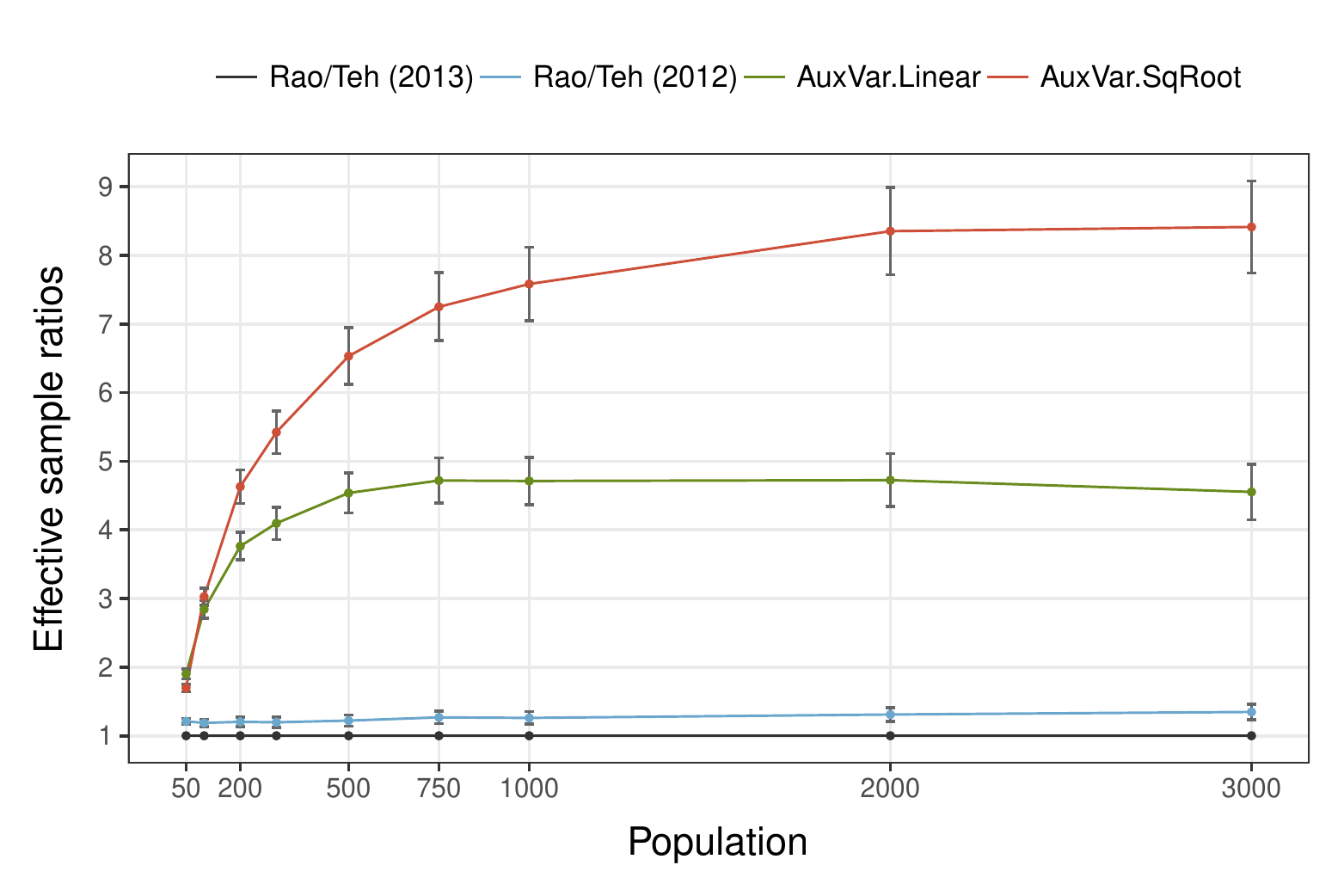}
\end{minipage}
\hfill
\begin{minipage}{0.34\textwidth}
\centering
\footnotesize
\vspace*{-15pt}
\renewcommand{\arraystretch}{1.4}
\setlength{\tabcolsep}{8pt}
\begin{tabular}{r|cc}
N    & Rate $\lambda$ & St. Dev. $\sigma$ \\ \hhline{===}
50   & 0.5            & 1.0               \\
100  & 1.0            & 2.0               \\
200  & 2.0            & 4.0               \\
300  & 3.0            & 6.0               \\
500  & 5.0            & 10.0              \\
750  & 7.5            & 15.0              \\
1000 & 10.0           & 20.0              \\
2000 & 20.0           & 40.0              \\
3000 & 30.0           & 60.0             
\end{tabular}
\end{minipage}
  \caption{Left, diagram with ratios in effective sample sizes versus a benchmark uniformization-based inference algorithm, for different algorithmic implementations. Intervals around points represent confidence intervals. Right, table summarizing the population levels, birth-rates and observation error used in simulations.}
  \label{fig:BDratios}
\end{figure}

Overall, Algorithm \ref{algoMainMJP} does not pose big gains over traditional uniformization, since rates for jumps in a birth-death system scale linearly with the population. Yet, from confidence metrics across both green and red scalings, we conclude that well tunned auxiliary-variable techniques presented in this section yield very significant efficiency gains; due to the approach naturally integrating within Gibbs steps for posterior paths of $X$.

\subsection{Example 2: Markovian stochastic epidemic models} 

Next, we address an inferential task with a time-homogeneous \textit{Susceptible-Infective-Removed} (SIR) stochastic epidemic model \citep{bailey1975mathematical}. Here, $X=(S_t,I_t,R_t)_{t\in[0,T]}$ tracks a population of $N$ individuals s.t. $\mathcal{S}=\{0,\dots,N\}^3$. At any time $t\in[0,T]$ each member of the population is either \textit{susceptible} (capable of contracting a disease), \textit{infective} (able to pass the disease to others) or \textit{removed} (immune to infection and unable to infect others). Since $S_t+I_t+R_t=N$, then $X\equiv (S_t,I_t)_{t\in[0,T]}$ corresponds to a bivariate jump process. In common applications, $X$ begins with a susceptible population of $N-1$ individuals, along with an infective member whose disease contraction time is unknown; the infinitesimal generator matrix $Q$ is s.t.
\begin{align*}
Q_{(s,i),(s',i')} &= 
  \begin{cases}
    \beta s i & \text{if} \quad s'=s-1, \, i'=i+1, \\
    \gamma i & \text{if} \quad s'=s, \, i'=i-1, 
  \end{cases}
\end{align*}
and $Q_{(s,i),(s',i')}=0$ otherwise. Therefore, infective individuals become removed (by death or recovery) after an independent infectious period with \textit{removal rate} $\gamma$. While infected, they may further transfer the disease to members of the susceptible population through Poisson contacts with \textit{infection rate} $\beta$. When the epidemic has ceased at some terminal time $T$, then the entire population is divided between susceptibles (who avoided infection) and removed members.

\noindent\textbf{ Observed removals and inference.} Here, $X$ can further be represented by a triplet $(\boldsymbol{t},\boldsymbol{s},\boldsymbol{i})$ of transition times along with corresponding susceptible/infected population vectors. In common inferential settings, all removal observations $\boldsymbol{t}^R$ are available; that is, some $k<N$ times $0\leq t^R_1 < \dots < t^R_k<T$ when infective individuals have either died or recovered from a disease. Thus,
\begin{align*}
\mathcal{L}(\boldsymbol{t}^R|\boldsymbol{t},\boldsymbol{s},\boldsymbol{i})&= \prod_{j=1}^{|\boldsymbol{t}^R|}\mathbb{I}(\lim_{t\nearrow t^R_j} I_{t} = I_{t^R_j} + 1), 
\end{align*} 
which is a simplified, analogue expression to \eqref{likJumps}, where removal \textit{jump observations} are always observed. The term $\mathcal{L}(\boldsymbol{t}^R|\beta,\gamma)$, key for inference tasks on the rates, requires integrating over a space of full infection and removal times with associated density
\begin{align*}
f_{X}(\boldsymbol{t},\boldsymbol{s},\boldsymbol{i}|\beta,\gamma)  & =  \pi(i_0) e^{-\sum_{j=0}^{|\boldsymbol{t}|} (\beta s_j i_j + \gamma i_j)(t_{j+1}-t_j)} \prod_{j=0}^{|\boldsymbol{t}|-1} (\beta s_j i_j)^{\mathbb{I}(i_{j+1} = i_j + 1)} (\gamma i_j)^{\mathbb{I}(i_{j+1} = i_j - 1)} ,
\end{align*}
where $t_{|\boldsymbol{t}|+1}=T$, and $\pi(i_0)$ denotes the distribution of the initial infection time. The expression offers a basis for an MCMC approach to inference, through augmentation of the MJP trajectory with missing infection times; in combination with samples from rate posteriors
$$f(\beta|\boldsymbol{t},\boldsymbol{s},\boldsymbol{i}) \propto \beta^{|\boldsymbol{t}^R|-1} e^{-\sum_{j=0}^{|\boldsymbol{t}|} \beta s_j i_j \cdot (t_{j+1}-t_j)}  \cdot \pi(\beta) \quad\text{and}\quad f(\gamma|\boldsymbol{t},\boldsymbol{s},\boldsymbol{i}) \propto \beta^{|\boldsymbol{t}^R|} e^{-\sum_{j=0}^{|\boldsymbol{t}|} \gamma i_j \cdot (t_{j+1}-t_j)}  \cdot \pi(\gamma)$$
and time-intervals $[t_0,T]$ with initial infection $\pi(t_0|t_1) \propto e^{-(\beta\cdot(N-1) + \gamma)(t_1-t_0) }\cdot \pi(t_0)$, where $\pi(\beta), \pi(\gamma)$ and $\pi(t_0)$ denote priors. This is usually achieved with Metropolis-Hastings steps \citep{o1999bayesian,jewell2009bayesian}, where updates proceed by proposing additions, deletions or moves of a proportion (usually half) of infection times; however, the scalability of the algorithm is reportedly poor. This is displayed in Figure \ref{fig:epidemOutput}, where we find output traces for parameters in a small epidemic ($N=50$). In all density and \textit{autocorrelation} (ACF) diagrams, black/grey representations correspond to an auxiliary-variable algorithm as introduced in this paper; red coloured counterparts relate to a benchmark M-H implementation \citep{o1999bayesian}. Severe efficiency differences may be observed within the ACF plot. On the left, we find posterior mean dynamics and a $\%95$ credible interval for $(I_t+R_t)_{t\in[0,T]}$; the dashed blue line corresponds to the real unobserved value, and the green line is the (observed) removal process $(R_t)_{t\in[0,T]}$ with jump times $\boldsymbol{t}^R$.
\begin{figure}[h!]
  \centering
   \includegraphics[width=\linewidth]{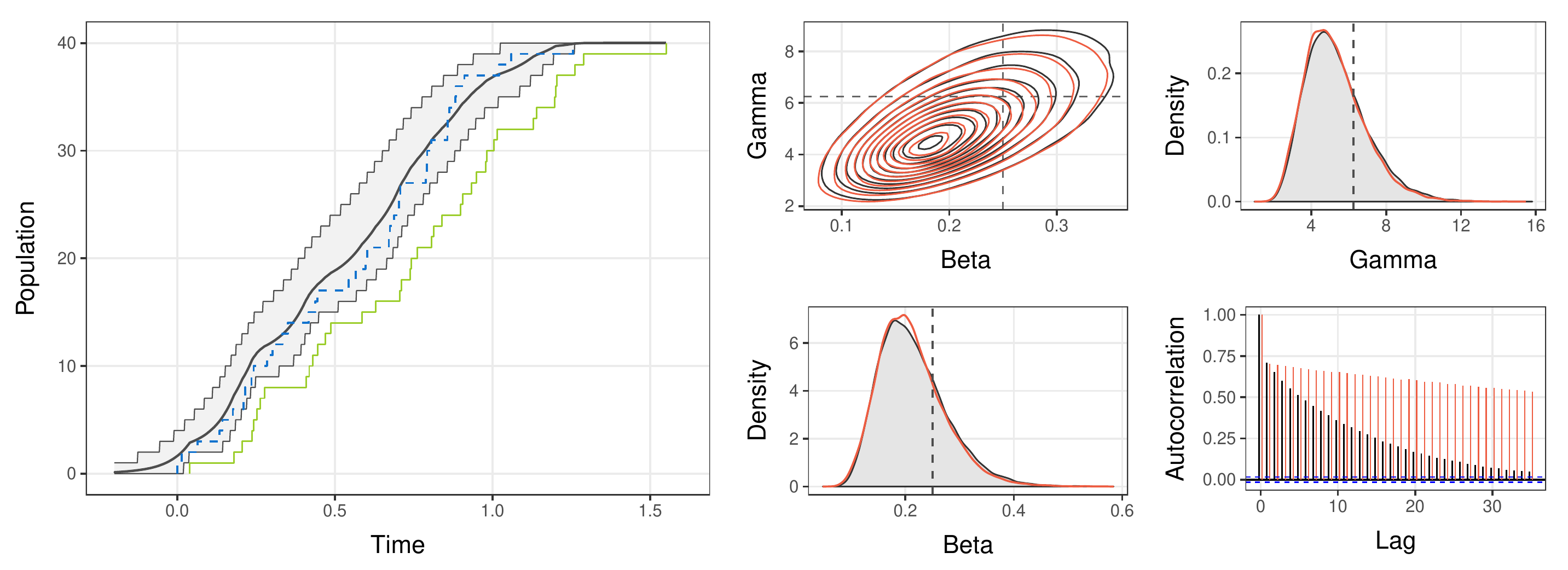}
  \caption{Epidemic study at capacity $N=50$.  Left, posterior mean dynamics and a $\%95$ credible interval for $(I_t+R_t)_{t\in[0,T]}$; the green line is the (observed) removal process $(R_t)_{t\in[0,T]}$. Right, density and \textit{autocorrelation} diagrams for an auxiliary-variable (black/grey) and Metropolis-Hastings (red) implementations.}
  \label{fig:epidemOutput}
\end{figure}

Next, we compare efficiency metrics in the procedures across increasing populations; and further analyze an adaptation of uniformization methods in \cite{rao13a} to system-jump observations (see definitions in Subsection \ref{twoSetpSubsec}). Removal data is simulated with rates $\gamma=1$, $\beta=2/N$ and securing a final removed population $R_T=N\cdot\%80$ (most representative outcome). Ratios on effective samples are reported within Figure \ref{fig:epidemESS}; there, benchmark lines (in blue) correspond to Algorithm \ref{naiveAlgoUnif} (not requiring exponential evaluations), with operator $\psi(x) = |Q_x/2|$, $x\in\mathcal{S}$ (candidate jumps attached to $\boldsymbol{t}$ with half intensity of diagonal in $Q$). We make this choice because (i) the model is stationary and (ii) existing alternative augmentation schemes do not scale (i.e. they do not work) with large populations. Also, (i) the green line represents the afore-mentioned benchmark epidemics Metropolis algorithm, (ii) the black line is for vanilla uniformization; with dominating rate $\Omega = 1.5 \cdot\max_{x\in\mathcal{S}} |Q_x|$, and (iii) the red line corresponds to sampling paths as deviations from mean-average dynamics $\xi(\cdot)$, given by the solution to a multivariate system
\begin{align*}
\frac{\mathrm{d}\xi_S(t)}{\mathrm{d}t} = -\beta \cdot \xi_S(t) \cdot \xi_I(t), \quad \frac{\mathrm{d}\xi_I(t)}{\mathrm{d}t} = \beta \cdot \xi_S(t) \cdot \xi_I(t) - \gamma\cdot \xi_I(t), \quad \text{and} \quad \frac{\mathrm{d}\xi_R(t)}{\mathrm{d}t} = \gamma\cdot \xi_I(t),
\end{align*}
with infection/removal parameters set to optimize $\min_{\beta,\gamma\in\mathbb{R}_+} \sum_{t_r\in\mathcal{D}[0,T]}(\xi_R(t_r)-R_{t_r})^2$ over an arbitrary discretization $\mathcal{D}[0,T]$ of the time interval. This is achieved incorporating Gamma variables in \eqref{gammaVars} over Algorithm \ref{naiveAlgoUnif}, with lag $l=25$, stationary mean $\mu=\log(N/10)$, autoregressive coefficient $\kappa=0.5$ and deviation $\sigma = 0.25$.
\begin{figure}[h!]
  \centering
   \includegraphics[width=\linewidth]{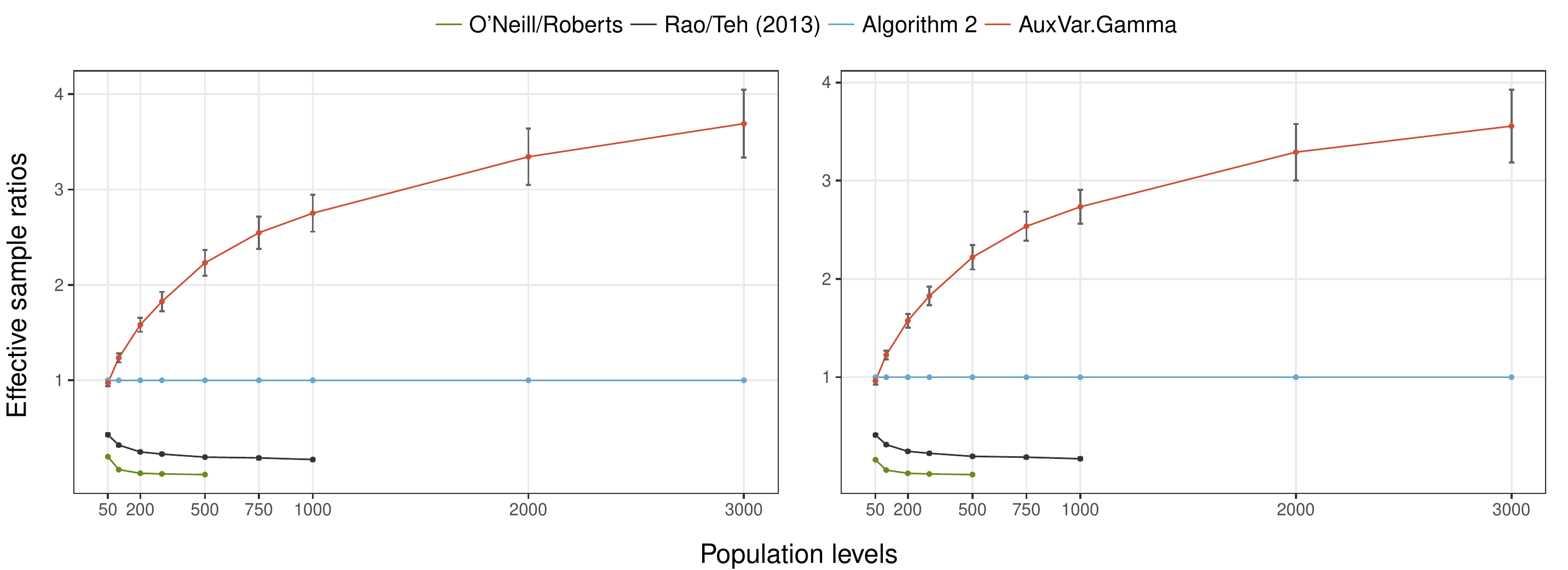}
  \caption{Ratios in effective sample sizes (with confidence intervals) against Algorithm \ref{naiveAlgoUnif} (in blue). Left diagram corresponds to infection rates; on the right, removal rates.}
  \label{fig:epidemESS}
\end{figure}

Existing inferential methods (green and black lines) do not scale to sizeable populations and perform poorly even within small ones. Noticeably, vanilla uniformization is bound to be inefficient in systems where generator rates scale quadratically; in epidemics, the data-augmentation procedure is associated with large dominating rates, often s.t. $\Omega > \beta \cdot (N/2)^2 + \gamma\cdot N$.

\section{Splitting the problem by mapping states or transitions} \label{splitting}

Finally, we discuss \textit{mappings} to reduce full MJP augmentations into families of smaller end-point conditioned tasks. A fixed $l\in\mathbb{N}$ will again define a \textit{lag} for auxiliary variables in \eqref{intAux},  among the discretization epochs in $\hat{\boldsymbol{t}}$. We thus employ a \textit{reduced} (deterministic) sequence $\{u_i\}_{i=l,2l,\dots}$ at times $\hat{t}_l,\hat{t}_{2l},\dots$ s.t. $u_i = \mathcal{T}(\hat{x}_{i-1}, \hat{x}_i)$ for some surjective mapping $\mathcal{T}:\mathcal{S}^2\rightarrow\mathcal{J}$; and variables in $\boldsymbol{u}$ are undefined other than for lagged times. Through $\mathcal{T}$, we map pairs of states in $\mathcal{S}^2$ to elements of the power set $\Sigma_\mathcal{S}$. A particular case of such construct was first discussed in \cite{Perez2017}; there, the authors simplify augmentation tasks for networked queueing systems by mapping MJP state transitions to job orderings across queues. Importantly, within the following examples, a lag $l$ must be (randomly) re-instantiated (or drifted) within every MCMC iteration, in order to ensure that trajectories $X$ are sampled from within their full support $\mathcal{X}$.

\noindent \textbf{Partitioning a state space.} Here, an (augmented) MJP process is forced to transition (small) population ranges at lagged times $\{\hat{t}_i\}_{i=l,2l,\dots}$. For a univariate $\mathcal{S}$-valued process example, we define 
$$\mathcal{J}\subset\Sigma_{\mathcal{S}} \quad \text{s.t.} \quad \varnothing\not\in\mathcal{J}, \, \cup_{A\in\mathcal{J}}=\mathcal{S} \: \text{and} \: A\cap B=\varnothing \, \text{for all} \, A,B\in\mathcal{J}.$$ 
Each \textit{part} $A\in\mathcal{J}$ must be defined s.t. jumps (including virtual self-jumps) restricted among its states yield an irreducible Markov chain on $A$. Then, for every existing (augmented) sequence $\hat{\boldsymbol{x}}$, we let $\mathcal{T}(x,x')=\{A\in\mathcal{J} : x' \in A\}$ map jumps $x\rightarrow x'$ at times $\{\hat{t}_i\}_{i=l,2l,\dots}$ to \textit{parts} $\{A_i\}_{i=l,2l,\dots}$ that contain the arrival states $\{\hat{x}_i\}_{i=l,2l,\dots}$. In order to re-sample a new \textit{compatible} sequence $\hat{\boldsymbol{x}}|\hat{\boldsymbol{t}},\boldsymbol{u}$ within Algorithms \ref{naiveAlgo}-\ref{algoMainMJP}, we can \textit{split} forward-backward procedures over intervals $[\hat{t}_{i-l},\hat{t}_{i})$, $i=l,2l,\dots$ s.t. each forward estimation 
$$\mathbb{P}(\hat{x}_{i}=x|u_0,\dots,u_{i-l};\hat{\boldsymbol{t}}) \quad \text{at epochs} \quad i=l,2l,\dots$$ 
is restricted to the subset $A_{i}$ of $\mathcal{S}$ and fed as the initial distribution $\pi(\cdot)$ at time $\hat{t}_{i}$ during the next interval. Backward steps proceed normally within and across sub-intervals. In Figure \ref{sampleCompatibles} (left) we find a sample diagram depicting this partitioning of the augmentation task. There, grey circles represent the reach of a univariate birth-death jump process at time points in $\hat{\boldsymbol{t}}$; blue squares (assigned at randomly lagged times, not equally spaced) correspond to ranges the process must transit. 
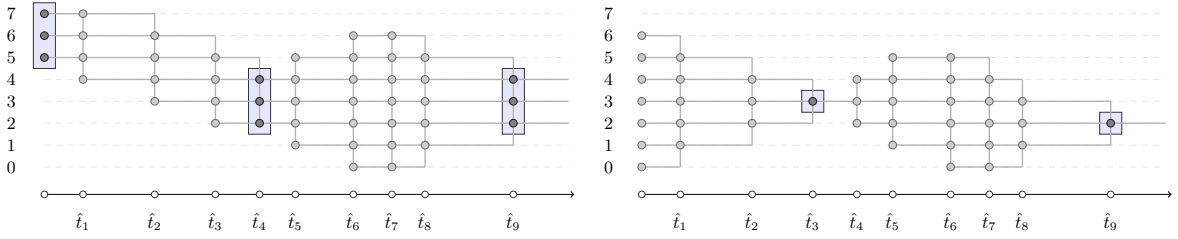
\begin{figure*}[h!]
\vskip 0.1in
\begin{center}
\resizebox{0.49\textwidth}{!}{%
\begin{tikzpicture}
\draw[->,line width=0.10mm] (0.1cm,-1.6cm) -- ++ (9.6,0cm);
\draw (0.1,-1.65cm) -- ++(0cm,0.1cm);
\foreach \i in {0,...,7} {
 \node at (-0.5,-1.1 + 0.4*\i) {\small $\i$}; \draw[dashed,draw=black!10!white,line width=0.10mm] (0.1,-1.1 + 0.4*\i) -- (9.6,-1.1 + 0.4*\i);
}
\fill[draw=black!80!white,fill=blue!10!white] (-0.1,0.7) rectangle (0.3,1.9);
\fill[draw=black!80!white,fill=blue!10!white] (3.8,-0.5) rectangle (4.2,0.7);
\fill[draw=black!80!white,fill=blue!10!white] (8.4,-0.5) rectangle (8.8,0.7);

\foreach \i in {5,...,7} {
 \draw[-,draw=black!30!white,line width=0.25mm] (0.1,-1.1 + 0.4*\i) -- ++ (0.7,0); \filldraw[fill=black!50!white,draw=black!90!white] (0.1,-1.1 + 0.4*\i) circle [radius=0.07cm]; 
}
\draw[-,draw=black!30!white,line width=0.25mm] (0.8,0.5) -- (0.8,1.7);

\foreach \i in {4,...,7} {
 \draw[-,draw=black!30!white,line width=0.25mm] (0.8,-1.1 + 0.4*\i) -- ++ (1.3,0); \filldraw[fill=black!20!white,draw=black!60!white] (0.8,-1.1 + 0.4*\i) circle [radius=0.07cm]; 
}
\draw[-,draw=black!30!white,line width=0.25mm] (2.1,0.1) -- (2.1,1.7);

\foreach \i in {3,...,6} {
 \draw[-,draw=black!30!white,line width=0.25mm] (2.1,-1.1 + 0.4*\i) -- ++ (1.1,0); \filldraw[fill=black!20!white,draw=black!60!white] (2.1,-1.1 + 0.4*\i) circle [radius=0.07cm]; 
}
\draw[-,draw=black!30!white,line width=0.25mm] (3.2,-0.3) -- (3.2,1.3);

\foreach \i in {2,...,5} {
 \draw[-,draw=black!30!white,line width=0.25mm] (3.2,-1.1 + 0.4*\i) -- ++ (0.8,0); \filldraw[fill=black!20!white,draw=black!60!white] (3.2,-1.1 + 0.4*\i) circle [radius=0.07cm]; 
}
\draw[-,draw=black!30!white,line width=0.25mm] (4,-0.3) -- (4,0.9);

\foreach \i in {2,...,4} {
 \draw[-,draw=black!30!white,line width=0.25mm] (4,-1.1 + 0.4*\i) -- ++ (0.65,0); \filldraw[fill=black!50!white,draw=black!90!white] (4,-1.1 + 0.4*\i) circle [radius=0.07cm]; 
}
\draw[-,draw=black!30!white,line width=0.25mm] (4.65cm,-0.7cm) -- (4.65cm,0.9cm);

\foreach \i in {1,...,5} {
 \draw[-,draw=black!30!white,line width=0.25mm] (4.65,-1.1 + 0.4*\i) -- ++ (1.05,0); \filldraw[fill=black!20!white,draw=black!60!white] (4.65,-1.1 + 0.4*\i) circle [radius=0.07cm]; 
}

\draw[-,draw=black!30!white,line width=0.25mm] (5.7cm,-1.1cm) -- (5.7cm,1.3cm);

\foreach \i in {0,...,6} {
 \draw[-,draw=black!30!white,line width=0.25mm] (5.7,-1.1 + 0.4*\i) -- ++ (0.7,0); \filldraw[fill=black!20!white,draw=black!60!white] (5.7,-1.1 + 0.4*\i) circle [radius=0.07cm]; 
}
\draw[-,draw=black!30!white,line width=0.25mm] (6.4cm,-1.1cm) -- (6.4cm,1.3cm);

\foreach \i in {0,...,6} {
 \draw[-,draw=black!30!white,line width=0.25mm] (6.4,-1.1 + 0.4*\i) -- ++ (0.6,0); \filldraw[fill=black!20!white,draw=black!60!white] (6.4,-1.1 + 0.4*\i) circle [radius=0.07cm]; 
}
\draw[-,draw=black!30!white,line width=0.25mm] (7,-1.1cm) -- (7,1.3cm);

\foreach \i in {1,...,5} {
 \draw[-,draw=black!30!white,line width=0.25mm] (7,-1.1 + 0.4*\i) -- ++ (1.6,0); \filldraw[fill=black!20!white,draw=black!60!white] (7,-1.1 + 0.4*\i) circle [radius=0.07cm]; 
}
\draw[-,draw=black!30!white,line width=0.25mm] (8.6,-0.7cm) -- (8.6,0.9cm);

\foreach \i in {2,...,4} {
 \draw[-,draw=black!30!white,line width=0.25mm] (8.6,-1.1 + 0.4*\i) -- ++ (1,0); \filldraw[fill=black!50!white,draw=black!90!white] (8.6,-1.1 + 0.4*\i) circle [radius=0.07cm]; 
}

\foreach \i in {0.1,0.8,2.1,3.2,4,4.65,5.7,6.4,7,8.6} {
\filldraw[fill=black!00!white,draw=black!80!white] (\i,-1.6) circle [radius=0.06cm]; 
}
\node at (0.8,-2.1) {\small $\hat{t}_1$};
\node at (2.1,-2.1) {\small $\hat{t}_2$};
\node at (3.2,-2.1) {\small $\hat{t}_3$};
\node at (4,-2.1) {\small $\hat{t}_4$};
\node at (4.65,-2.1) {\small $\hat{t}_5$};
\node at (5.7,-2.1) {\small $\hat{t}_6$};
\node at (6.4,-2.1) {\small $\hat{t}_7$};
\node at (7,-2.1) {\small $\hat{t}_8$};
\node at (8.6,-2.1) {\small $\hat{t}_9$};

\end{tikzpicture}}\hfill
\resizebox{0.49\textwidth}{!}{%
\begin{tikzpicture}
\draw[->,line width=0.10mm] (0.1cm,-1.6cm) -- ++ (9.6,0cm);
\draw (0.1,-1.65cm) -- ++(0cm,0.1cm);
\foreach \i in {0,...,7} {
 \node at (-0.5,-1.1 + 0.4*\i) {\small $\i$}; \draw[dashed,draw=black!10!white,line width=0.10mm] (0.1,-1.1 + 0.4*\i) -- (9.6,-1.1 + 0.4*\i);
}

\fill[draw=black!80!white,fill=blue!10!white] (3,-0.1) rectangle (3.4,0.3);
\fill[draw=black!80!white,fill=blue!10!white] (8.4,-0.5) rectangle (8.8,-0.1);

\foreach \i in {0,...,6} {
 \draw[-,draw=black!30!white,line width=0.25mm] (0.1,-1.1 + 0.4*\i) -- ++ (0.7,0); \filldraw[fill=black!20!white,draw=black!60!white] (0.1,-1.1 + 0.4*\i) circle [radius=0.07cm]; 
}
\draw[-,draw=black!30!white,line width=0.25mm] (0.8,-1.1) -- (0.8,1.3);

\foreach \i in {1,...,5} {
 \draw[-,draw=black!30!white,line width=0.25mm] (0.8,-1.1 + 0.4*\i) -- ++ (1.3,0); \filldraw[fill=black!20!white,draw=black!60!white] (0.8,-1.1 + 0.4*\i) circle [radius=0.07cm]; 
}
\draw[-,draw=black!30!white,line width=0.25mm] (2.1,-0.7) -- (2.1,0.9);

\foreach \i in {2,...,4} {
 \draw[-,draw=black!30!white,line width=0.25mm] (2.1,-1.1 + 0.4*\i) -- ++ (1.1,0); \filldraw[fill=black!20!white,draw=black!60!white] (2.1,-1.1 + 0.4*\i) circle [radius=0.07cm]; 
}
\draw[-,draw=black!30!white,line width=0.25mm] (3.2,-0.3) -- (3.2,0.5);

\foreach \i in {3,...,3} {
 \draw[-,draw=black!30!white,line width=0.25mm] (3.2,-1.1 + 0.4*\i) -- ++ (0.8,0); \filldraw[fill=black!50!white,draw=black!90!white] (3.2,-1.1 + 0.4*\i) circle [radius=0.07cm]; 
}
\draw[-,draw=black!30!white,line width=0.25mm] (4,-0.3) -- (4,0.5);

\foreach \i in {2,...,4} {
 \draw[-,draw=black!30!white,line width=0.25mm] (4,-1.1 + 0.4*\i) -- ++ (0.65,0); \filldraw[fill=black!20!white,draw=black!60!white] (4,-1.1 + 0.4*\i) circle [radius=0.07cm]; 
}
\draw[-,draw=black!30!white,line width=0.25mm] (4.65cm,-0.7cm) -- (4.65cm,0.9cm);

\foreach \i in {1,...,5} {
 \draw[-,draw=black!30!white,line width=0.25mm] (4.65,-1.1 + 0.4*\i) -- ++ (1.05,0); \filldraw[fill=black!20!white,draw=black!60!white] (4.65,-1.1 + 0.4*\i) circle [radius=0.07cm]; 
}
\draw[-,draw=black!30!white,line width=0.25mm] (5.7cm,-1.1cm) -- (5.7cm,0.9cm);

\foreach \i in {0,...,5} {
 \draw[-,draw=black!30!white,line width=0.25mm] (5.7,-1.1 + 0.4*\i) -- ++ (0.7,0); \filldraw[fill=black!20!white,draw=black!60!white] (5.7,-1.1 + 0.4*\i) circle [radius=0.07cm]; 
}
\draw[-,draw=black!30!white,line width=0.25mm] (6.4cm,-1.1cm) -- (6.4cm,0.9cm);

\foreach \i in {0,...,4} {
 \draw[-,draw=black!30!white,line width=0.25mm] (6.4,-1.1 + 0.4*\i) -- ++ (0.6,0); \filldraw[fill=black!20!white,draw=black!60!white] (6.4,-1.1 + 0.4*\i) circle [radius=0.07cm]; 
}
\draw[-,draw=black!30!white,line width=0.25mm] (7,-1.1cm) -- (7,0.5cm);

\foreach \i in {1,...,3} {
 \draw[-,draw=black!30!white,line width=0.25mm] (7,-1.1 + 0.4*\i) -- ++ (1.6,0); \filldraw[fill=black!20!white,draw=black!60!white] (7,-1.1 + 0.4*\i) circle [radius=0.07cm]; 
}
\draw[-,draw=black!30!white,line width=0.25mm] (8.6,-0.7cm) -- (8.6,0.1cm);

\foreach \i in {2,...,2} {
 \draw[-,draw=black!30!white,line width=0.25mm] (8.6,-1.1 + 0.4*\i) -- ++ (1,0); \filldraw[fill=black!50!white,draw=black!90!white] (8.6,-1.1 + 0.4*\i) circle [radius=0.07cm]; 
}

\foreach \i in {0.1,0.8,2.1,3.2,4,4.65,5.7,6.4,7,8.6} {
\filldraw[fill=black!00!white,draw=black!80!white] (\i,-1.6) circle [radius=0.06cm]; 
}
\node at (0.8,-2.1) {\small $\hat{t}_1$};
\node at (2.1,-2.1) {\small $\hat{t}_2$};
\node at (3.2,-2.1) {\small $\hat{t}_3$};
\node at (4,-2.1) {\small $\hat{t}_4$};
\node at (4.65,-2.1) {\small $\hat{t}_5$};
\node at (5.7,-2.1) {\small $\hat{t}_6$};
\node at (6.4,-2.1) {\small $\hat{t}_7$};
\node at (7,-2.1) {\small $\hat{t}_8$};
\node at (8.6,-2.1) {\small $\hat{t}_9$};

\end{tikzpicture}}
\vskip 0in
\caption{Toy diagrams with compatible birth-death paths given auxiliary variables at lagged times. Greyed circles represent the reach of the jump process across marginal time points. Left, a partitioning the augmentation task by defining ranges to transition; right, analogue technique using mappings to states.} 
\label{sampleCompatibles}
\end{center}
\vskip -0.1in
\end{figure*}

\noindent \textbf{Sampling end-point conditioned bridges.} To further simplify data-augmentation, the above partition may be defined s.t. $\mathcal{J}=S$ and $\mathcal{T}(x,x')=x'$, for all $x,x'\in\mathcal{S}$. Thus, a new \textit{compatible} sequence $\hat{\boldsymbol{x}}|\hat{\boldsymbol{t}},\boldsymbol{u}$ will be \textit{locked} at times $\{\hat{t}_i\}_{i=l,2l,\dots}$, and forward-backward procedures are independent across subintervals $[\hat{t}_{i-l},\hat{t}_{i})$. The approach is depicted within Figure \ref{sampleCompatibles} (right), where our algorithms will sample bridges across the auxiliary \textit{mapped} states.

In both cases, methods easily generalise to multivariate process settings, and trajectories may straightforwardly be conditioned on data by following previously introduced conventions. Noticeably, the correlation across subsequent trajectory samples for $X$ will be drastically increased; yet, this is compensated by considerably simplified procedures, and we below report on efficiency results with algorithmic implementations for a predator-prey model. C++ repositories to reproduce these results may be found on \href{https://github.com/IkerPerez/scalableSamplingMJPs}{github.com/IkerPerez/scalableSamplingMJPs}.

\subsection{Example 3: An stochastic Lotka-Volterra model}

A \textit{Lotka-Volterra} model \citep{boys2008bayesian} describes predator-prey interactions among two biological species. Here, a non-stationary process $X=(X^1_{t},X^2_{t})_{t\geq 0}$  evolves stochastically according to rates
\begin{align*}
Q_{(x_1,x_2),(x_1+1,x_2)}(t) = \alpha(t) \cdot x_1, &\quad Q_{(x_1,x_2),(x_1-1,x_2)}(t) = \beta(t) \cdot x_1 \cdot x_2, \\
Q_{(x_1,x_2),(x_1,x_2+1)}(t) = \delta(t) \cdot x_1 \cdot x_2, &\quad Q_{(x_1,x_2),(x_1,x_2-1)}(t) = \gamma(t) \cdot x_2.
\end{align*}
Thus, $(X^1_{t})_{t\geq 0}$ refers to the prey population and $(X^2_{t})_{t\geq 0}$ is the predator counterpart. Within the following inferential task, functionals decompose between interaction parameters and a seasonality modifier; so that $\alpha(t)=\alpha \cdot r(t)$, $\beta(t)=\beta\cdot r(t)$ and so on, for some (known) $r(t)\in[1,2], \, t\geq 0$. Additionally, an initial state is (uniformly) randomized between (bounded) populations with capacity $N\in\mathbb{N}_0$.

\noindent\textbf{State measurements and inference.} For simulated datasets (at various population bounds), we produce noisy state observations s.t. $O_r\sim\mathcal{N}(X_{t_r},N/25)$ at equally spaced times $t_r\in[0,T]$, $r\geq 1$. Throughout, parameter choices $\alpha=0.125$, $\beta=\delta=0.005$ and $\gamma = 0.1$ are assigned, with $r(t) = 3/2 + cos(2\pi \cdot t/T)/2$. Similarly to previous examples, backwards inference on the rates proceeds by data augmentation of trajectory densities $f_{X}(\boldsymbol{t},\boldsymbol{x}|\boldsymbol{Q})$ in \eqref{pathProbs}, along with draws from the posterior $f(\alpha,\beta,\delta,\gamma|\boldsymbol{t},\boldsymbol{x})$ (which factors across the individual rates). In Figure \ref{fig:preypred} we find a sample representation of augmented prey and predator population trajectories (red lines), along with observations (dark circles). There, dashed lines represent posterior mean-average paths, and grey areas are for $\%95$ credible intervals.
\begin{figure}[h!]
  \centering
   \includegraphics[width=\linewidth]{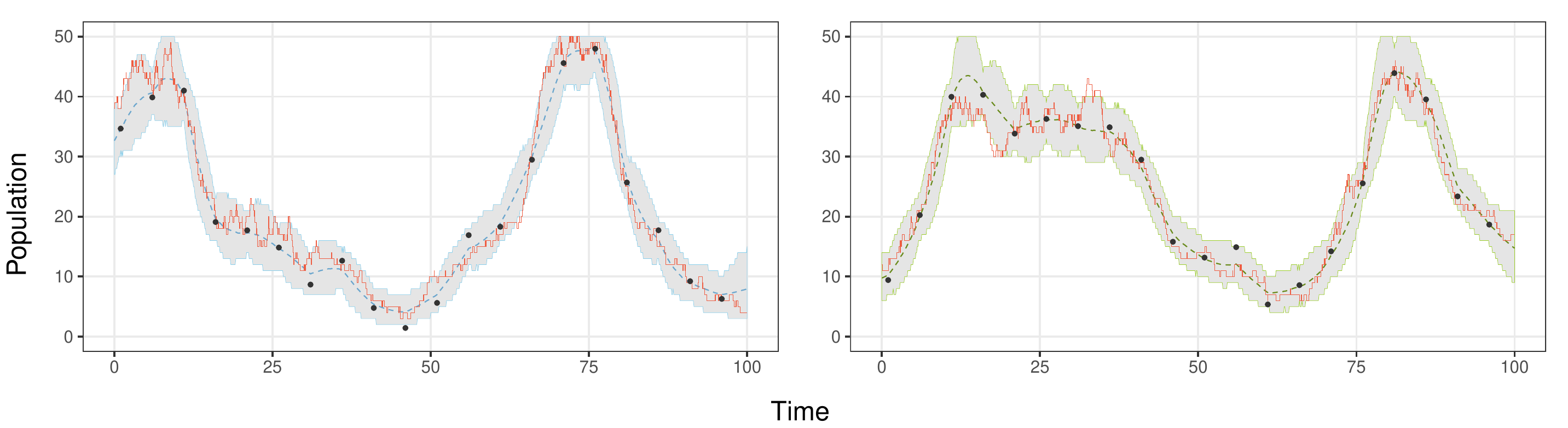}
  \caption{In red, augmented prey (left) and predator (right) population paths. Observations are represented by dark circles. Dashed lines and greyed areas are for posterior mean-average paths and $\%95$ credible intervals.}
  \label{fig:preypred}
\end{figure}

Efficiency results comparing different augmentation methods are shown in Figure \ref{fig:BDratios}. As before, the diagrams display ratios (along with confidence intervals) in effective sample sizes (scaled for computation time). The horizontal axes represent bounds imposed over \textit{each} marginal biological species; thus, the real explorable state space tested increases up to $120^2=14.400$. The left diagram corresponds to \textit{average} effective samples across parameter rates $\alpha, \beta, \delta, \gamma$; instead, the right diagram represents ratios on the \textit{minimum} effective samples across the $4$ parameters. In both instances, the reference line (in red) at level $1$ corresponds to sampling end-point conditioned \textit{bridges} across (randomised) intervals with lag $l=0.5\cdot N$, built on top of Algorithm \ref{algoMainMJP} with operator $\psi(t,x) = 0.5 \cdot |Q_x(t)|$, $(t,x)\in[0,T]\times\mathcal{S}$. The green line is for the same procedure, but using a lag $l=0.75\cdot N$; and the blue line represents a plain implementation of Algorithm \ref{algoMainMJP} without auxiliary variables driving an increase in efficiency. Finally, in black we observe efficiency results for a vanilla uniformization procedure with dominating rate $\Omega = 1.5 \cdot\max_{x\in\mathcal{S}} \sup_{t\in[0,T]}|Q_x(t)|$. 
\begin{figure}[h!]
\centering
\includegraphics[width=\linewidth]{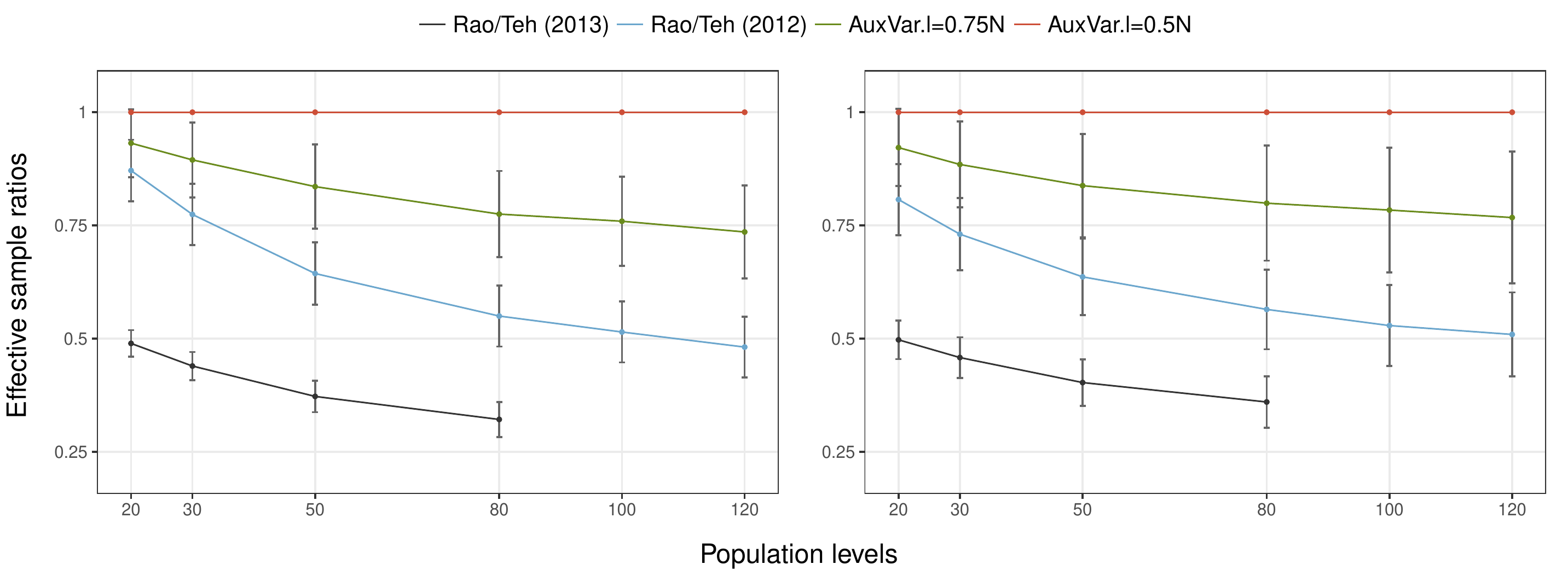}
  \caption{Ratios in effective sample sizes versus a benchmark auxiliary-variable procedure. The left diagram corresponds to mean sample sizes across $\alpha, \beta, \delta, \gamma$; on the right, equivalent ratios for minimum sample sizes.}
  \label{fig:BDratios}
\end{figure}

Results are consistent with auxiliary-variable methods introduced in Section \ref{ODEsection}. In all cases, the various alternatives introduced in this paper can (i) scale inferential uniformization-based inferential frameworks to much larger problems, and (ii) drive significant increases in computational efficiency.

\section{Discussion}

This paper has presented a novel and comprehensive framework for the design of scalable data-augmentation procedures, suitable for use within \textit{exact} Bayesian inferential tasks, and applicable to birth-death, epidemic or predator-prey systems, to name only a few. The need for auxiliary-variable augmentation designs as presented here is justified by the limitations in existing state-of-the-art uniformization-based approaches \citep[see][and references therein]{hobolth2009,rao2012mcmc,rao13a,miasojedow2015particle,georgoulas2017unbiased,zhang2017efficient}, which are inefficient, unadaptable or unusable with mid-sized or large population systems, often associated with multiple types of observational data.
 
We have reported on results that apply multiple MCMC algorithm construction to problems of broad statistical interest, and demonstrated prior claims on efficiency and scalability benefits, by direct comparison to current benchmark methods in the literature. Finally, since the presented framework builds on uniformized representations of non-stationary jump processes, we note that the various techniques introduced in this paper will be only applicably to purely Markovian processes.

\bibliographystyle{apa}
\setlength{\bibsep}{2pt} 
\renewcommand{\bibfont}{\small}
\bibliography{bibliography}

\begin{thebibliography}{}

\bibitem[\protect\astroncite{Bailey et~al.}{1975}]{bailey1975mathematical}
Bailey, N.~T. et~al. (1975).
\newblock {\em The mathematical theory of infectious diseases and its
  applications}.
\newblock Charles Griffin \& Company Ltd.

\bibitem[\protect\astroncite{Boys et~al.}{2008}]{boys2008bayesian}
Boys, R.~J., Wilkinson, D.~J., and Kirkwood, T.~B. (2008).
\newblock Bayesian inference for a discretely observed stochastic kinetic
  model.
\newblock {\em Statistics and Computing}, 18(2):125--135.

\bibitem[\protect\astroncite{Capp{\'e} et~al.}{2003}]{cappe2003reversible}
Capp{\'e}, O., Robert, C.~P., and Ryd{\'e}n, T. (2003).
\newblock Reversible jump, birth-and-death and more general continuous time
  markov chain monte carlo samplers.
\newblock {\em Journal of the Royal Statistical Society: Series B (Statistical
  Methodology)}, 65(3):679--700.

\bibitem[\protect\astroncite{Daley and
  Vere-Jones}{2007}]{daley2007introduction}
Daley, D.~J. and Vere-Jones, D. (2007).
\newblock {\em An introduction to the theory of point processes: volume II:
  general theory and structure}.
\newblock Springer Science \& Business Media.

\bibitem[\protect\astroncite{Fearnhead and Sherlock}{2006}]{fearnhead2006exact}
Fearnhead, P. and Sherlock, C. (2006).
\newblock An exact gibbs sampler for the markov-modulated poisson process.
\newblock {\em Journal of the Royal Statistical Society: Series B (Statistical
  Methodology)}, 68(5):767--784.

\bibitem[\protect\astroncite{Georgoulas et~al.}{2017}]{georgoulas2017unbiased}
Georgoulas, A., Hillston, J., and Sanguinetti, G. (2017).
\newblock Unbiased bayesian inference for population markov jump processes via
  random truncations.
\newblock {\em Statistics and computing}, 27(4):991--1002.

\bibitem[\protect\astroncite{Gillespie}{1977}]{gillespie1977exact}
Gillespie, D.~T. (1977).
\newblock Exact stochastic simulation of coupled chemical reactions.
\newblock {\em The journal of physical chemistry}, 81(25):2340--2361.

\bibitem[\protect\astroncite{Golightly and
  Sherlock}{2018}]{golightly2018efficient}
Golightly, A. and Sherlock, C. (2018).
\newblock Efficient sampling of conditioned markov jump processes.
\newblock {\em arXiv preprint arXiv:1809.07139}.

\bibitem[\protect\astroncite{Golightly and
  Wilkinson}{2015}]{golightly2015bayesian}
Golightly, A. and Wilkinson, D.~J. (2015).
\newblock Bayesian inference for markov jump processes with informative
  observations.
\newblock {\em Statistical applications in genetics and molecular biology},
  14(2):169--188.

\bibitem[\protect\astroncite{Gross et~al.}{2008}]{Gross:2008:FQT:1972549}
Gross, D., Shortle, J.~F., Thompson, J.~M., and Harris, C.~M. (2008).
\newblock {\em Fundamentals of Queueing Theory}.
\newblock Wiley-Interscience, New York, NY, USA, 4th edition.

\bibitem[\protect\astroncite{Higdon}{1998}]{HigdonAux}
Higdon, D.~M. (1998).
\newblock Auxiliary variable methods for markov chain monte carlo with
  applications.
\newblock {\em Journal of the American Statistical Association},
  93(442):585--595.

\bibitem[\protect\astroncite{Hobolth and Stone}{2009}]{hobolth2009}
Hobolth, A. and Stone, E.~A. (2009).
\newblock Simulation from endpoint-conditioned, continuous-time markov chains
  on a finite state space, with applications to molecular evolution.
\newblock {\em The Annals of Applied Statistics}, 3(3):1204--1231.

\bibitem[\protect\astroncite{Jensen}{1953}]{jensen1953markoff}
Jensen, A. (1953).
\newblock Markoff chains as an aid in the study of {M}arkoff processes.
\newblock {\em Scandinavian Actuarial Journal}, 36:87--91.

\bibitem[\protect\astroncite{Jewell et~al.}{2009}]{jewell2009bayesian}
Jewell, C.~P., Kypraios, T., Neal, P., and Roberts, G.~O. (2009).
\newblock Bayesian analysis for emerging infectious diseases.
\newblock {\em Bayesian Analysis}, 4(3):465--496.

\bibitem[\protect\astroncite{Kurtz}{1970}]{kurtz_1970}
Kurtz, T.~G. (1970).
\newblock Solutions of ordinary differential equations as limits of pure jump
  markov processes.
\newblock {\em Journal of Applied Probability}, 7(1):49–58.

\bibitem[\protect\astroncite{Kurtz}{1971}]{kurtz_1971}
Kurtz, T.~G. (1971).
\newblock Limit theorems for sequences of jump markov processes approximating
  ordinary differential processes.
\newblock {\em Journal of Applied Probability}, 8(2):344–356.

\bibitem[\protect\astroncite{Miasojedow and
  Niemiro}{2015}]{miasojedow2015particle}
Miasojedow, B. and Niemiro, W. (2015).
\newblock Particle gibbs algorithms for markov jump processes.
\newblock {\em arXiv preprint arXiv:1505.01434}.

\bibitem[\protect\astroncite{Miasojedow et~al.}{2017}]{miasojedow2017geometric}
Miasojedow, B., Niemiro, W., et~al. (2017).
\newblock Geometric ergodicity of rao and teh’s algorithm for markov jump
  processes and ctbns.
\newblock {\em Electronic Journal of Statistics}, 11(2):4629--4648.

\bibitem[\protect\astroncite{O'Neill and Roberts}{1999}]{o1999bayesian}
O'Neill, P.~D. and Roberts, G.~O. (1999).
\newblock Bayesian inference for partially observed stochastic epidemics.
\newblock {\em Journal of the Royal Statistical Society: Series A (Statistics
  in Society)}, 162(1):121--129.

\bibitem[\protect\astroncite{Opper and
  Sanguinetti}{2008}]{opper2008variational}
Opper, M. and Sanguinetti, G. (2008).
\newblock Variational inference for markov jump processes.
\newblock In {\em Advances in Neural Information Processing Systems}, pages
  1105--1112.

\bibitem[\protect\astroncite{Perez and Casale}{2018}]{perez2018approximate}
Perez, I. and Casale, G. (2018).
\newblock Approximate bayesian inference with queueing networks and coupled
  jump processes.
\newblock {\em arXiv preprint arXiv:1807.08673}.

\bibitem[\protect\astroncite{Perez et~al.}{2018}]{Perez2017}
Perez, I., Hodge, D., and Kypraios, T. (2018).
\newblock Auxiliary variables for bayesian inference in multi-class queueing
  networks.
\newblock {\em Statistics and Computing}, 28(6):1187--1200.

\bibitem[\protect\astroncite{Rao and Teh}{2012}]{rao2012mcmc}
Rao, V. and Teh, Y.~W. (2012).
\newblock {MCMC} for continuous-time discrete-state systems.
\newblock In {\em Advances in Neural Information Processing Systems}, pages
  701--709.

\bibitem[\protect\astroncite{Rao and Teh}{2013}]{rao13a}
Rao, V.~A. and Teh, Y.~W. (2013).
\newblock Fast {MCMC} sampling for {M}arkov jump processes and extensions.
\newblock {\em Journal of Machine Learning Research}, 14:3295--3320.

\bibitem[\protect\astroncite{Sutton and Jordan}{2011}]{sutton2011}
Sutton, C. and Jordan, M.~I. (2011).
\newblock Bayesian inference for queueing networks and modeling of internet
  services.
\newblock {\em The Annals of Applied Statistics}, 5(1):254--282.

\bibitem[\protect\astroncite{Van~Dijk}{1992}]{van1992uniformization}
Van~Dijk, N.~M. (1992).
\newblock Uniformization for nonhomogeneous markov chains.
\newblock {\em Operations research letters}, 12(5):283--291.

\bibitem[\protect\astroncite{Van~Dijk et~al.}{2018}]{van2018uniformization}
Van~Dijk, N.~M., Van~Brummelen, S. P.~J., and Boucherie, R.~J. (2018).
\newblock Uniformization: Basics, extensions and applications.
\newblock {\em Performance evaluation}, 118:8--32.

\bibitem[\protect\astroncite{Zhang and Rao}{2018}]{zhang2017efficient}
Zhang, B. and Rao, V. (2018).
\newblock Efficient parameter sampling for markov jump processes.
\newblock {\em arXiv preprint arXiv:1704.02369}.

\end{thebibliography}

\end{document}